\documentclass[english,a4paper]{amsart}

\usepackage{enumerate}
\usepackage{amsmath}
\usepackage{epsfig}
\usepackage{graphics}
\usepackage[normalem]{ulem}
\usepackage[usenames, dvipsnames]{color}
\usepackage{multirow}
\usepackage{xr}
\newcommand{\frachalf}{\frac{1}{2}}
\newcommand{\fracthird}{\frac{1}{3}}
\newcommand{\fractwothird}{\frac{2}{3}}

\theoremstyle{plain}
\newtheorem{thm}{Theorem}

\newtheorem{lemma}[thm]{Lemma}

\newcommand{\be}{\begin{equation}}
\newcommand{\ee}{\end{equation}}

\title{Coupled fast and slow feedbacks lead to continual evolution: A general modeling approach}

\author{Meike T. Wortel}
\author{Han Peters}
\author{Nils Chr. Stenseth}

\address{Meike Wortel and Nils Chr. Stenseth:
Centre for Ecological and Evolutionary Synthesis (CEES),
 Department of Biosciences, University of Oslo, Oslo, Norway
 }
\address{Han Peters:
Korteweg de Vries Institute for Mathematics\\
 University of Amsterdam, Amsterdam, The Netherlands}
\address{Present address Meike Wortel:
Institite for Biodiversity and Ecosystem Dynamics\\
University of Amsterdam, Amsterdam, The Netherlands}

\address{meike.t.wortel@gmail.com}

\date{}

\begin{document}

\begin{abstract}
The Red Queen Hypothesis, which suggests that continual evolution can result from solely biotic interactions, has been studied in macroevolutionary and microevolutionary contexts. While the latter has been effective in describing examples in which evolution does not cease, describing which properties lead to continual evolution or to stasis remains a major challenge. In many contexts it is unclear which assumptions are necessary for continual evolution, and whether described behavior is robust under perturbations. Our aim here is to prove continual evolution under minimal conditions and in a general framework, thus automatically obtaining robustness. 
We show that the combination of a fast positive and a slow negative feedback causes continual evolution with a single evolving trait, provided the ecological timescale is sufficiently separated from the timescales of mutations and negative feedback. Our approach and results form a next step towards a deeper understanding of the evolutionary dynamics resulting from biotic interactions.
\end{abstract}

\maketitle

\section{Introduction}

Species typically live in complex ecosystems with many interactions among them and external drivers. The evolutionary dynamics of a species in a complex ecosystem can be caused by the properties of the species, by the interaction with the coexisting species or by external drivers. To better understand to what extent the emerging ecological and evolutionary dynamics is caused by intrinsic properties of the species or  biotic interactions (within- and between species interaction) we need to ignore abiotic drivers. Such an abiotically unchanging environment may lead to a static adaptive landscape, where the adaptation will follow a path towards a peak in that landscape---reachable or not \cite{wiser2013long}. Since the major part of any individual's environment typically is composed of other (evolving) species, any species' environment will in general change even without external abiotic variation, both through ecological and evolutionary changes. Hence, the adaptive landscape will be dynamic, which can lead to continual co-evolutionary dynamics, where traits over evolutionary time are fluctuating, escalating or chasing each other \cite{brockhurst2014running}.

Of particular interest to us is the emergence of the so-called Red Queen Dynamics \cite{van1973new}, a concept that has had a major influence on micro- and macroevolutionary theory. \cite{stenseth1984coevolution} and others \cite{vermeij2013reining,liow2011red,voje2015role} have analyzed a macro-evolutionary model aiming at understanding under what conditions continual Red Queen dynamics and stasis results from within-system biotic interactions: Stenseth and Maynard Smith demonstrated that both could result depending upon the nature of the within system biotic interactions --- without being able translated into ecological terms what these conditions were. \cite{nordbotten2016asymmetric} found that symmetric competitive interactions are more likely to lead to stasis. Another approach is studying interactions of a few species in more mechanistic detail, microevolutionary Red Queen (RQ) dynamics \cite{brockhurst2014running}.

Theoretical studies of micro-evolutionary RQ dynamics mostly use methods based on adaptive dynamics and quantitative genetics \cite{dieckmann1995evolutionary,mougi2010evolution,marrow1992coevolution,marrow1993evolutionary,dieckmann1996dynamical,branco2018eco}. The adaptive dynamics approach assumes the ecological dynamics have reached an equilibrium and studies the invasion of individuals with a slightly deviating phenotype of the adaptive trait. The advantage of this approach is that it allows for a rigid theoretical analysis of the system. The quantitative genetics approach does not assume the ecological dynamics to be in equilibrium, but has a timescale separation between the evolutionary adaptation and the ecological dynamics. If the evolutionary rate of change is very slow, the quantitative genetics approach becomes similar to the adaptive dynamics approach. Both methods assume that adaptive traits evolve along a fitness gradient. Studies focussing on predator-prey and host-parasite systems have been able to reach conclusions about conditions that increase or decrease the chance of RQ dynamics in a specific setting (e.g. fast adaptation is less likely to lead to RQ dynamics \cite{mougi2010evolution} and RQ dynamics requires an intermediate harvesting efficiency of the prey \cite{dieckmann1995evolutionary}).

Most of the above mentioned studies use specific functional forms for their analysis, hampering the generalisation of the obtained results. Meta-analysis (such as the one by \cite{abrams2000evolution}) can provide some more general insights, but the conclusions are still limited, especially since many studies use similar equations. To obtain general results and therefore a broad understanding of what ecological interactions can cause certain evolutionary patterns, such as continual evolution, we need as general models as possible. With this contribution we aim at extending the theoretical understanding of under which conditions continual evolution and stasis will result. For this purpose we use a very general model with few assumptions regarding the form of the model functions. With such a general approach (making a minimum of assumption on the functional forms and the parameters) we increase the robustness of the obtained results. We find that a system with slow and fast feedback interactions exhibits continual Red Queen type of dynamics depending on the timescales. Moreover we allow for a polymorphic population, not constraining the distribution of phenotypes that may be present in the population, and mutations of small and large effects.

\section{Model description and the emerging eco-evolutionary dynamics}

\subsection{Model description}
In order to focus our argument, we use a general model of an evolving trait for a single species, which can be extended to multiple species:

\begin{equation}
\begin{aligned}
\frac{du}{dt} & = u \cdot f(u, \varphi) + \epsilon_m \cdot g(u)\\
\frac{d\varphi}{dt} & = \epsilon_e \cdot h(u, \varphi).
\end{aligned}
\end{equation}

Here $u$ represents the population density distribution over the trait space. In the general case the evolving trait is not specified and therefore its trait space could be anything. The external factors (abiotic factors, species without evolving traits) are captured in $\varphi$. The function $f(u, \varphi)$ describes the growth of the population (which can depend on the distribution in trait space $u$ as well as the external factors $\varphi$). The function $g(u)$ describes the change in the trait due to mutations. The function $g$ does not incorporate how well these phenotypes perform (the change due to mutations is modeled in an unbiased manor), the change of the trait to more fit individuals comes from the modelling of the growth of all phenotypes. The function $h(u, \varphi)$ describes the change in external factors, which can depend on the phenotype densities $u$. Time scale differences between growth, external factors and mutations are captured in $\epsilon_e$ and $\epsilon_m$.

We include a possible polymorphic trait distribution (as is preferable argued by \cite{law1997evolution} and also used in \cite{van1995predator}). A polymorphic trait distribution arises easily with asexual reproduction or traits that are determined by a few loci, but can also arise when assortive mating develops (see the discussion in \cite{kisdi2002red}). We use an approach including genetic variation by modeling the populations with differential equations, which also releases the assumption that mutations have to cause an infinitesimal change in the trait value.

An example of a slow negative feedback is a predator prey interaction, where the generation time of the predator is much slower than that of the prey and the predator predates preferable on prey with a specific phenotype. An example of the specificity of the predator is predators preferring a certain size of prey (e.g. bears preferring large salmon \cite{carlson2011eco} or zooplankton preferring phytoplankton with a certain nitrogen:phosphorus ratio \cite{branco2018eco}). Other examples besides predator-prey interactions include: a herbivore that changes the vegetation which then becomes less suitable for the herbivore; micro-organisms that produce a compound that inhibits them but is diluted in a large volume (to achieve the feedback being slow); a species with a complex life cycle with a habitat shift (with an evolvable component), if a crowded habitat has a delayed effect on the habitat quality; or a more complicated response of the whole ecosystem that leads to a negative feedback (e.g. with intransitive cycles as in \cite{bonachela2017eco}). The last example includes cases where the feedback is caused by human intervention, such as flu vaccinations---common viruses will be vaccinated against during the next season.

The fast positive feedback is an Allee effect on the phenotype. This could be due to for example finding mates (if the trait has effect on suitability of mates), cooperation in defense or in feeding. An example is when the adaptive trait is foraging on a certain vegetation type, the positive feedback could be that the more individuals forage in the same place, the better they are protected against predation. Or in the case of a habitat shift in the life cycle, similar habitat shift will share the second habitat with more individuals, providing mates or protection.

\subsection{A single species with two phenotypes shows continual evolutionary dynamics when the negative feedback is sufficiently slow}

The best way to understand the system is to reduce it to the simplest version. Therefore we start with an simple version of a trait with only two possible phenotypes, one that is interacting with a external feedback and one that is not (or less strongly) (see Figure \ref{fig:binaryTheory}A). Although the main objective of this simplification is to understand the mechanisms that lead to continual evolution, this representation also reflects a biological system where a trait is either present or not at all, e.g. choosing one host or another. When the trait also affects mate choice (e.g. because they will be located near the same host), the trait has a positive feedback. The slow negative feedback results from the fact that one of the hosts can develop defenses.

We analyse the system where both phenotypes and the feedback change over time. We assume that the phenotypes have a positive feedback effect, i.e. the more individuals with a certain phenotype, the higher the growth rate of that phenotype. To be precise, we study the following system with the interacting phenotype $A$ (the phenotype choosing the host that can develop defenses), the (less) interacting phenotype $B$ (the phenotype choosing the host that cannot develop defenses) and the feedback compound $\varphi$ (the amount of defense in the host population that can develop defenses in the example):

\begin{equation}
\begin{aligned}
\frac{dA}{dt} & = A \cdot f_A(A,B, \varphi) + \epsilon_m \cdot g(A,B)\\
\frac{dB}{dt} & = B \cdot f_B(A,B, \varphi) - \epsilon_m \cdot g(A,B)\\
\frac{d\varphi}{dt} & = \epsilon_e \cdot h(A, B, \varphi).
\end{aligned} \label{eq:binary}
\end{equation}

The growth functions of the phenotypes $f_A$ and $f_B$ include the positive feedback. $\varphi$ is increased more by $A$ than by $B$ which is incorporated in the function $h$. Mutations are possible from one phenotype to the other with the mutation function $g$. $\epsilon_m$ and $\epsilon_e$ denote the different timescales of the population dynamics (no $\epsilon$), the feedback dynamics ($\epsilon_e$) and the evolutionary dynamics ($\epsilon_m$). We have modeled a positive effect of $A$ on $\varphi$ and a negative effect of $\varphi$ on $A$, but all our results are also valid for the opposite, then we can substitute $\varphi$ by $-\varphi$.

We prove that, for the system in Figure \ref{fig:binaryTheory}A and Equations \eqref{eq:binary} and with only slight restrictions on the functions $f$, $g$ and $h$, if we can choose $\epsilon_e$ and $\epsilon_m$ sufficiently small, the system will exhibit continual evolution for almost all initial states. The detailed proof is in the Supplementary Information section \ref{sec:twophenofixedpopsize} and \ref{sec:twophenovariablepopsize}; the reasoning is as follows: First, we prove the case of a constant population size, where we can describe our system by the feedback compound $\varphi$ and the ratio $(R)$ between the two phenotypes. Observing that the proof for this case does not rely upon knowing the exact value of $\frac{dR}{dt}$ but merely the \emph{sign}, we can extend the proof to systems where the population size is not fixed, by adopting an assumption controlling the sign of $\frac{dR}{dt}$. There are several possibilities but we used the 'unique stable value assumption': when we keep $\varphi$ fixed and know $A$ there is only one value of $B$ that guarantees a stable state, and likewise when we know $B$ there is only one value of $A$. The proof only requires knowledge of the sign in a subregion of the domain. 

\begin{figure}[!htbp]
  \centering
  \includegraphics[width=.95 \textwidth]{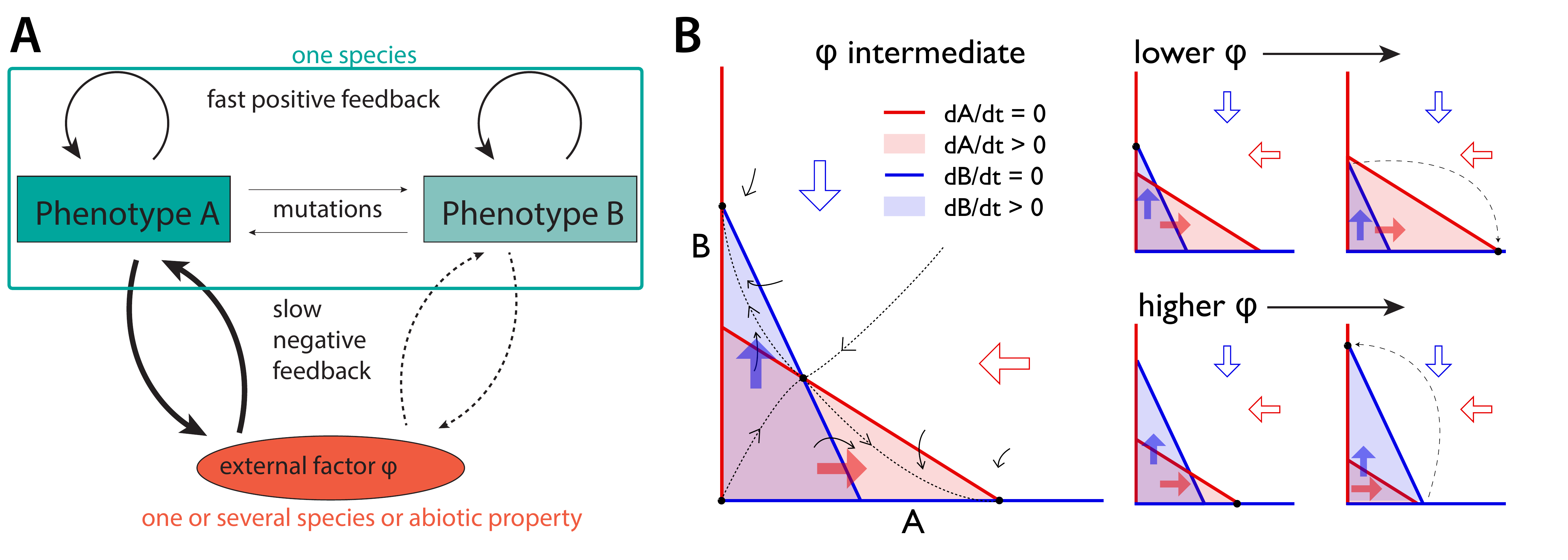}
  \caption{\textbf{Continual evolution with two phenotypes.} \textbf{A} A species has phenotypes $A$ and $B$, which both have fast positive feedbacks on their own phenotype. Phenotype $A$ has a strong negative feedback with an external factor $\varphi$ (e.g. an external compound or other species such as a predator), while phenotype $B$ has no, or a weaker, negative feedback. Mutations are possible between the two phenotypes, although rates are low. \textbf{B} The assumptions guarantee that the phase plane for intermediate $\varphi$ resemble the one shown on the left. The isoclines for $dA/dt=0$ and $dB/dt=0$ cross as shown because the fast positive feedback ensures that the growth of phenotype $B$ decreases faster with increasing $A$ than the growth of phenotype $A$ and vice versa. The system will either tend to the intersection of the $dA/dt$ and $dB/dt$ isoclines on the $A$ or the $B$ axis. Near the $B$ axis $\varphi$ will decrease and as shown on the top right the phase plane diagram will change such that the intersection will become unstable ($dA/dt$ becomes positive) and the system will go to the stable equilibrium at the intersection at the $A$ axis. Similarly, close to the $A$ axis $\varphi$ will increase, and the intersection near the $A$ axis will become unstable. These properties will lead to continual evolution.}
  \label{fig:binaryTheory}
\end{figure}

The idea of the proof for a variable population size is illustrated in Figure \ref{fig:binarySystem}B. The diagrams show an approximation of the dynamics in the ($A$,$B$) phase planes. This approximation is without the mutation term, but since mutations are rare (in our case the mutation term is small), the actual phase plane is very similar to the one shown. When at some point the feedback level is intermediate, there are 4 intersections of the lines $\frac{dA}{dt}=0$ and $\frac{dB}{dt}=0$ of which one is repelling (the origin), one is a saddle point (the point in the middle) and the other two are attracting (the nature of the equilibria can be seen from the sign of the derivatives). Depending on the initial densities of $A$ and $B$ the system goes to an attracting point near the $A$-axis or the $B$-axis. When this is near the $A$-axis, $B$ will be almost 0 and since $A$ is interacting strongly with the feedback compound, $\varphi$ will increase. This will in turn change the phase plane diagram and the intersection near the $A$-axis will become repelling. Then there is only one attracting intersection and that is near the $B$-axis and the system will approach that state. Since the dynamics in $\varphi$ are slower than in $A$ and $B$, the system will come close to the intersection near the $B$-axis before $\varphi$ changes significantly. When $\varphi$ changes we are back to the initial figure for intermediate $\varphi$, but now $\varphi$ will continue to decrease and the intersection near the $B$-axis will disappear. This will continue indefinitely and therefore will lead to continual evolution.

\begin{figure}[!htbp]
  \includegraphics[width=.5 \textwidth]{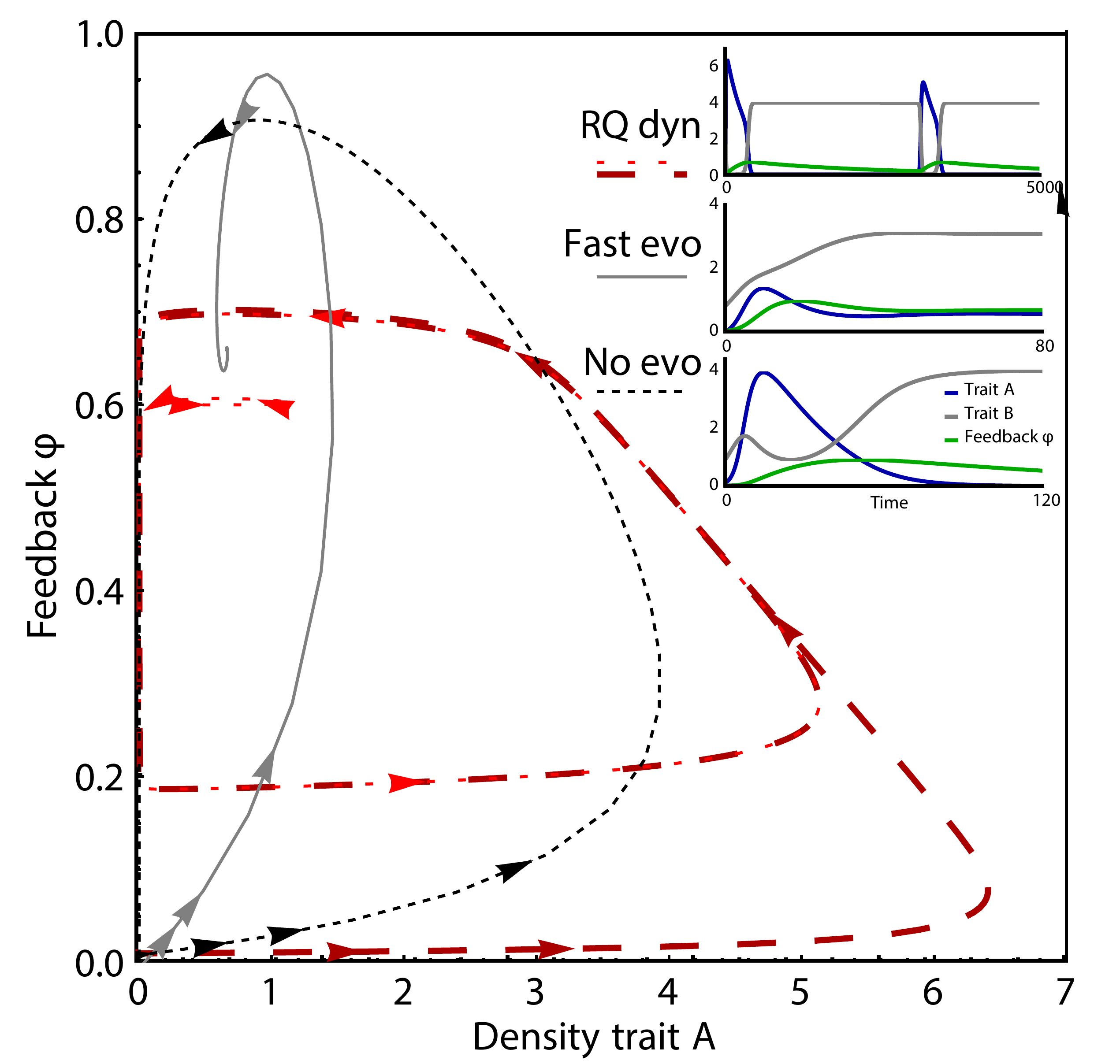}
  \caption{\textbf{Example system with two phenotypes.} Behaviour of a system following the outline of Fig \ref{fig:binaryTheory}A. Phase plane diagram of the density of individuals with phenotype A and the feedback $\varphi$ shows the possible system behaviour depending on the relative timescales of population dynamics, the feedback and mutations. Slow feedback and mutations ($\epsilon_m = 0.00005$ and $\epsilon_e = 0.0005$) leads to continual cyclic evolutionary dynamics (Red Queen dynamics, red lines and top inset), almost independent of initial conditions (2 initial conditions are shown). With fast feedback and mutations ($\epsilon_m = 0.01$ and $\epsilon_e = 0.1$) an equilibrium is reached (gray line and middle inset). No evolution ($\epsilon_m = 0$ and $\epsilon_e = 0.01$) leads to the extinction of one of the traits (black dashed line and bottom inset). Equations are given in the Supplementary Material section \ref{sec:example_twophenovariablepopsize}.}
  \label{fig:binarySystem}
\end{figure}

We only need some restrictions on the functions $f$, $g$ and $h$ and do not use any functional forms in our proof. Figure \ref{fig:binarySystem} illustrates the simulation of an example system. We show continual evolutionary cycles for 2 different initial states if $\epsilon_e$ and $\epsilon_m$ are chosen sufficiently small. With a high mutation rate and/or a fast feedback a stable equilibrium will be reached (leading to stasis), while lack of mutations causes one of the two phenotypes to go extinct.

\subsection{A multi-valued trait}

For studying the case of a multi-valued trait, the evolving trait is represented as a range of values, and the external factor is represented by a scalar quantity (normalised between 0 and 1). We can then describe the system with the following set of differential equations:

\begin{equation}
\begin{aligned}
\frac{du_i}{dt} & = u_i \cdot f(u, \varphi) + \epsilon_m \cdot g(u)\\
\frac{d\varphi}{dt} & = \epsilon_e \cdot h(u, \varphi).
\end{aligned}
\end{equation}

We split the population in discrete groups $u_i$ with similar trait values. We show that the assumptions of a fast positive feedback and a slow negative feedback on the phenotype in this otherwise general model lead to continual evolution. These assumptions are given by restrictions on the functions $f$ and $h$. We can extend the reasoning from the previous section to a multi-valued trait. With a multi-valued trait we mean a trait with more than two phenotypes or an approximation of a continuous trait, such as length or level of toxin production. An overview of the system is shown in Figure \ref{fig:continuousTheory}A, where the trait values are scaled from 0 to 1 according to the strength of their negative feedback. We call these phenotypes $i$, and the density of individuals with phenotype $i$ is denoted by $u_i$. The feedback compound is again called $\varphi$ and can be the same quantities as in the previous section. For the positive feedback we use the population average of the trait, which we call $M$. We assume that when the level of $\varphi$ is so low that phenotypes with higher $i$ have higher growth rates for a certain phenotype distribution, growth rates will keep increasing with $i$ when we increase $M$ (and vice versa, if $\varphi$ is high and growth rates decrease with $i$, a decrease in $M$ will not change this). We use the assumption that the growth rate is either strictly increasing or strictly decreasing in the trait (therefore one of the extreme traits always has the highest fitness).

\begin{figure}
  \includegraphics[width=\textwidth]{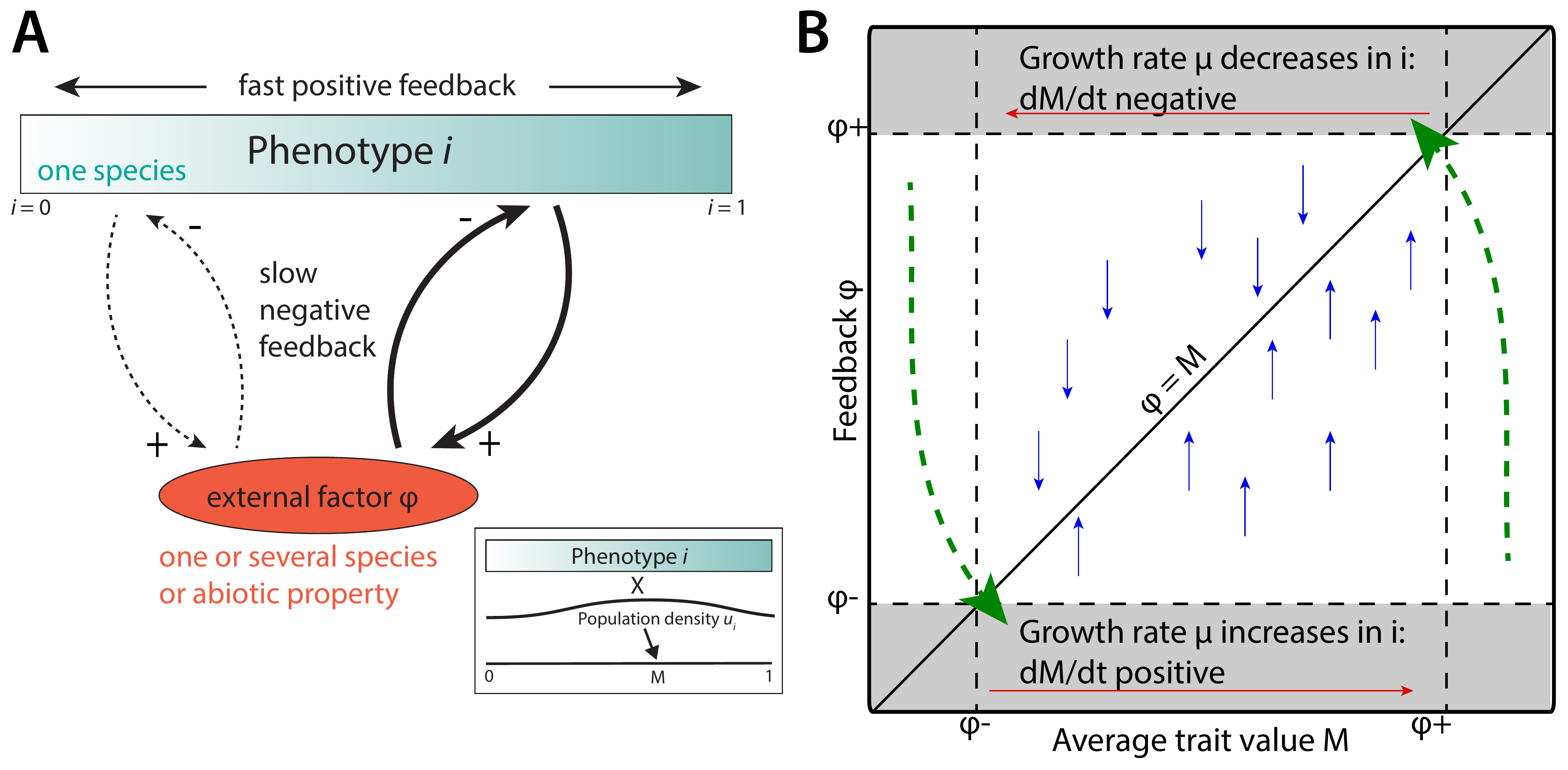}
  \caption{\textbf{System with a many-valued trait.} \textbf{A} One species has a multi-valued trait, with a phenotype ranging from $i=0$ to $i=1$. Higher values of $i$ denote a stronger interaction with the negative feedback $\varphi$, which can be a species, an abiotic factor or a property of the ecosystem. The positive feedback is implemented as a high average trait value $M$ leading to higher growth rates for larger trait values. The inset shows that $M$ is calculated by taking the average trait value in the population, thus weighing the trait values by the density of that phenotype. \textbf{B} Outline of the proof in the ($\varphi$, $M$) plane. $\varphi$ always goes to $M$ (blue arrows), while M quickly increases when $\varphi$ is small and quickly decreases when $\varphi$ is large (red arrows). The combination of these effects leads to continual oscillating evolutionary dynamics (green arrows).}
  \label{fig:continuousTheory}
\end{figure}

Analogous to the system with two phenotypes, we can prove the existence of evolutionary fluctuations under conditions mimicking our earlier assumptions. In the multiple phenotype setting there is one subtle point in the proof, where an additional assumption is needed to guarantee relatively quick changes in phenotypes. One possible additional assumption guaranteeing such rapid fluctuations is the existence of mutations from any trait value to any other trait value. In practical settings this assumptions will not always be applicable. For example, the assumption is more realistic when the trait value describes enzyme levels or levels of toxin production, than when considering length of a species, as we expect length to evolve gradually. An alternative assumption will be discussed later.

The proof that we can choose the timescales to obtain continual evolutionary dynamics is outlined in Fig. \ref{fig:continuousTheory}B, and described in detail in the Supplementary Information section \ref{sec:multivaluedtrait}. Here we provide an outline: We scale $\varphi$ also between 0 and 1 and without loss of generality may assume that $\varphi$ tends to $M$. Then, when for example $M$ is high and $\varphi$ is low, $\varphi$ will increase (blue arrows in Fig. \ref{fig:continuousTheory}B) because $\varphi$ goes to $M$. At some point $\varphi$ and $M$ will cross, in which case $M$ must already be declining and will continue to decline. The assumed fast positive feedback ensures that $M$ decreases drastically in a time interval in which $\varphi$ hardly changes, leading to symmetric conditions and hence continual evolutionary oscillations (green arrows in Fig. \ref{fig:continuousTheory}B).

\begin{figure}[!htbp]
  \includegraphics[width= \textwidth]{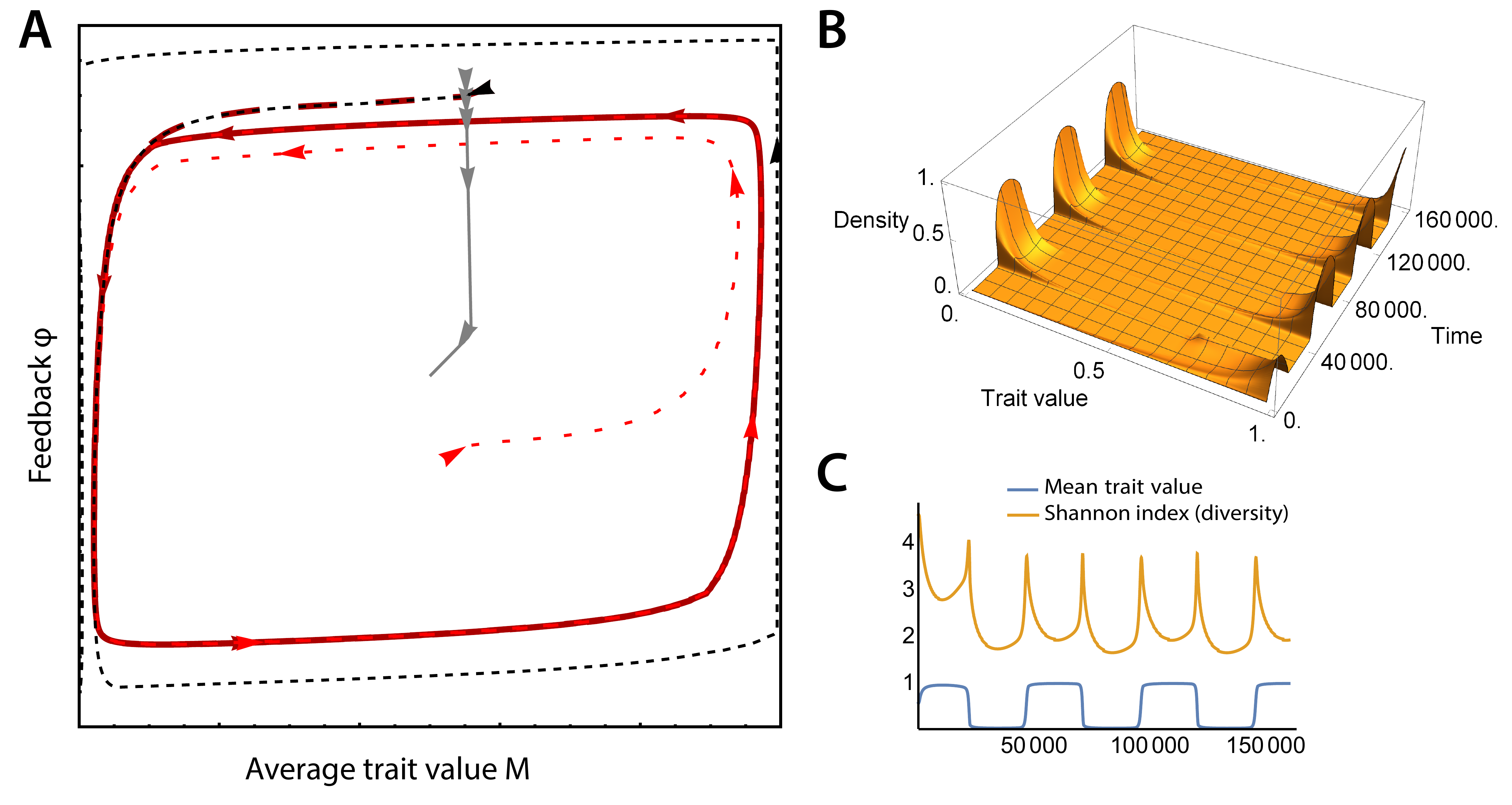}
  \caption{\textbf{Dynamics of a system with a multi-valued trait.} Behaviour of a system following the outline of Figure \ref{fig:continuousTheory}. \textbf{A} System behaviour for no evolution (thin dashed lines), slow evolution (red and solid dashed lines; RQ dynamics from different intital values) and fast evolution (gray dashed line, leading to an equilibrium). \textbf{B} Phenotype abundances from time simulations of the system with RQ dynamics. \textbf{C} Mean trait value and Shannon diversity (a diversity measure) show the oscillations. Simulation of 100 phenotypes equally distributed over the range of trait values (40 phenotypes for the phenotype abundances). See Supplementary Information section \ref{sec:multivaried} for the equations and parameters.}
  \label{fig:continuousSystem}
\end{figure}

Again our proof is context free and applies to a wide variety of equations; an example is given in Figure \ref{fig:continuousSystem}. In Figure \ref{fig:continuousSystem}A demonstrates that turning off evolution leads to extinction of all but one phenotype, fast feedback leads to an equilibrium, and slow feedback leads to continual fluctuations. Figure \ref{fig:continuousSystem}B shows the distribution of phenotypes over time, and in \ref{fig:continuousSystem}C it is shown that the diversity of phenotypes increases at shifts of the mean trait value.

If we do not want to assume the existence of mutations from any phenotype to any other phenotype, we can adopt an alternative assumption, namely that the negative feedback is chosen sufficiently slow \emph{depending on} the rate of mutations. In some circumstances this may even mean that the rate at which $\varphi$ changes is even slower than the mutation rate. This alternative assumption again may or may not be desirable in practice.

\subsection{Extensions to multiple feedbacks}

In natural systems, multiple traits determine the fitness in an individual. An extension of our model to more traits will give a more realistic picture of the dynamics we expect in natural populations. We extended the system to a species with two different traits that are both associated with a different negative feedback. The number of possible phenotypes is now the product of the number of phenotypes for each of the traits. In the most simple case, there are four phenotypes in the population: (1) high susceptibility for both feedbacks, (2) low susceptibility for both feedbacks, and (3 and 4) high susceptibility for one and low for the other (Figure \ref{fig:TwoFeedbacks}A). Simulations with this extended system show that this system can also lead to continual adaptation and that the dynamics generally become irregular (Figure \ref{fig:TwoFeedbacks}B). Looking at the change of the diversity index over time, we can conclude that there is no fixed period in the dynamics (Figure \ref{fig:TwoFeedbacks}C). The example with multiple traits shows that we do not expect a natural system to come back to exactly the same state in a regular fashion.

\begin{figure}[!htbp]
  \includegraphics[width= \textwidth]{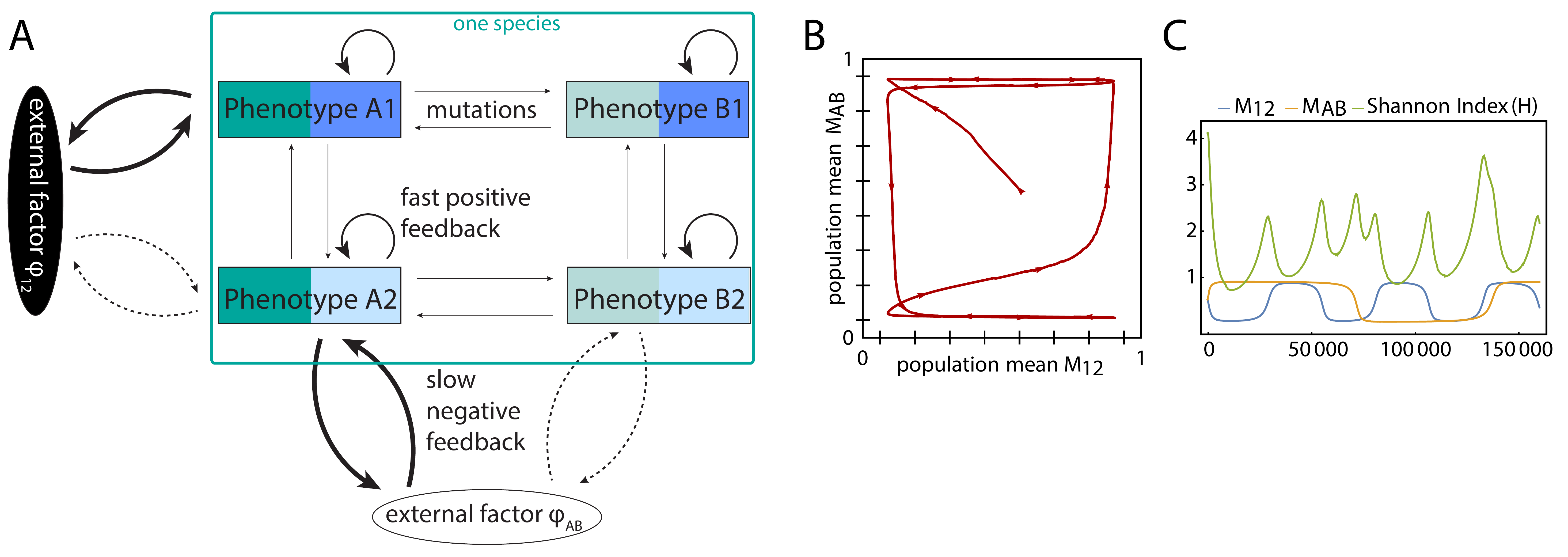}
  \caption{\textbf{Species with two traits associated with two different feedbacks.} \textbf{A} There are two external feedbacks ($\varphi_{1}$ and $\varphi_{A}$) and in the simplest case two phenotypes associated with each feedback, leading to 4 phenotypes in total. More phenotypes per trait are also possible, leading to a total number of phenotypes as the product of the phenotypes for each trait. \textbf{B} The change in mean population phenotype ($M$) for both traits ($i_1$ and $i_A$) in time for an example with more phenotypes per trait. Ticks on the axes show possible individual phenotypes used in the simulation.  \textbf{C} Mean trait values and diversity over time. With more than one trait periods are not as regular as can be seen from the non-periodic Shannon Index ($H$). Equations and parameters for the figure are given in the Supplementary Material section \ref{sec:example_twofeedbacks}.}
  \label{fig:TwoFeedbacks}
\end{figure}

\section{Discussion}

We have shown that a simple motif of fast positive feedback and slow negative feedback leads to continual evolution. Instead of working with a specific model we have shown that continual evolution occurs for other similar models. We allow for phenotypic variation in the population and large effect mutations. We automatically show robustness to model form, because the evolutionary dynamics do not depend on specific equations or parameter values, and also robustness to small fluctuations. Large fluctuations might affect the dynamics, but because our dynamics are reachable from almost all initial conditions, the system will return to continual evolution after a large fluctuation. Previously, it has been shown that continual evolution can be found in specific simple models. In the literature robustness is indicated by showing that a range of parameter values or a different equation will lead to the same result \cite{dercole2006coevolution,mougi2010evolution}. Here we try to avoid the use of specific equations and parameter values to obtain more general results and stronger robustness. Previous examples of papers not using specific functional forms consider the possibility of the occurrence of evolutionary sliding \cite{dercole2006coevolution}, the conclusion that symmetric interactions are more likely to lead to stasis \cite{nordbotten2016asymmetric} and the Red Queen evolution in bacterial communities with cyclic inhibition in three species \cite{bonachela2017eco}, although all of these papers include some restrictions on the models. \cite{dercole2006coevolution} focusses only on slow-fast models with separated timescales, \cite{nordbotten2016asymmetric} only consider bilinear species interactions and \cite{bonachela2017eco} only consider a fixed set of three species that form a non-transitive cycle with their species interactions.

Our results are accompanied by specific examples, but the results are not restricted to those examples. In all our examples the continual evolutionary dynamics stem from switching between multiple ecological attractors that are steady states (ecogenetically driven Red Queen Dynamics mode B\cite{khibnik1997three}). Our conclusions remain the same when these attractors are limit cycles instead of steady states. An example of this type of dynamics is given in \cite{khibnik1997three} Figure 4 and, interestingly, directly using our assumptions on this example removes the evolutionary oscillations (only ecological cycles remain, with ecologically driven Red Queen Dynamics---that is fluctuations in traits are fast and small and only follow the ecological population density dynamics). However, when we make the negative feedback slower by decreasing the parameters for the predator dynamics tenfold (parameters r4 and $\gamma$), we retrieve the ecogenetically driven Red Queen Dynamics (results not shown).

Some of our results are in line with and extend previous conclusions. The dynamics of the continual evolution are very similar to microevolutionary Red Queen dynamics; the difference is that continual evolution does not require co-evolution. \cite{khibnik1997three} mention that Red Queen dynamics with a single evolving trait are possible when the dynamics are ecologically driven (the traits follow ecological dynamics) or ecogenetically driven and switching between two different ecological attractors. Continual evolution of a single trait might be quite prevalent, since examples include predator-prey systems in which the prey evolves but the predator does not (or much slower; as might be relatively common, see e.g. \cite{vermeij1982unsuccessful}). Moreover, our two trait example reproduces the result that fast adaptation is less likely to lead to RQ dynamics \cite{mougi2010evolution}. A new observation is that for a multi-valued trait, evolution should not be too slow relative to the negative feedback, and that ecology and evolutionary dynamics can influence each other. The rationale behind this is that the average population trait has to change drastically while the level of the negative feedback variable changes relatively little. If a trait evolved gradually that means that evolution should not be too much slower than the dynamics of the feedback, because then the feedback would 'catch up' with the average population trait while this average is near the equilibrium value. An experimental example of such kind of dynamics is that in cases such as a predator-prey system with a slow predator, the non-optimal predator phenotypes might stay in the population and not go extinct. In that case the species does not need to get a de novo mutations, but can get the new phenotype through selection and recombination. An example is a laboratory system of an algae and a rotifer \cite{yoshida2003rapid}. This type of RQ dynamics where the traits stay in the population will be more common in higher organisms and is an example of evolutionary rescue. Faster evolutionary change through mutations might also occur when species are regularly exposed to different environments and evolve adaptability \cite{crombach2007chromosome,crombach2008evolution}.

\begin{figure}[!htbp]
  \includegraphics[width= \textwidth]{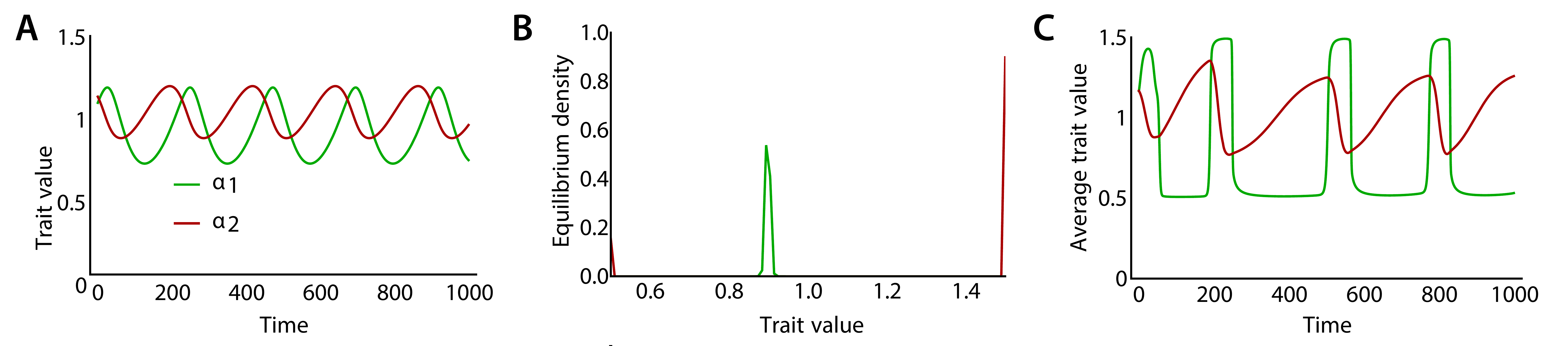}
  \caption{\textbf{RQ dynamics disappears when adding polymorphism to a population.} \textbf{A} The competitor-competitor model of \cite{khibnik1997three} shows the evolvable trait ($\alpha$) of the two competitors following each other in a RQ manner (figure replicated from figure 2 in \cite{khibnik1997three}). \textbf{B} When we allow for a polymorphic population the system tends to a bimodal distribution for the trait $\alpha_2$. \textbf{C} When we change the timescales (we made the dynamics of competitor two ten times slower), the RQ dynamics return. For equations and parameters see \ref{sec:literature_models}.}
  \label{fig:OtherModels}
\end{figure}

Phenotypic variation is often left out of Red Queen dynamics analysis. Outright restricting phenotypic variation can have a pronounced effect on the results. We have adapted a competitor-competitor model (Figure \ref{fig:OtherModels}A) with originally an adaptive dynamics approach that previously showed Red Queen Dynamics to an instance with our assumptions (the possibility of a polymorphic trait distribution and larger mutations). As shown in Figures \ref{fig:OtherModels}B, switching to the modelling approach used in this paper removes the RQ dynamics that were found in the original paper. However, when we change the timescales and make the feedback (competitor species in this example) slower, we retrieve the feedback, in line with the results of this paper (see Figure \ref{fig:OtherModels}C). In \cite{mougi2010evolution}, the RQ dynamics remain with phenotypic variation (data not shown), but there the RQ dynamics are ecologically driven: the adaptation follows the ecological predator-prey cycles. Other methods that do include phenotypic varions are the model by \cite{dieckmann1995evolutionary}, which does contain a polymorphic model which, different from our model, is also stochastic; models of evolutionary branching and extinction using adaptive dynamics (summarized in \cite{kisdi2002red}); and \cite{van1995predator}, where a Lotka-Volterra predator-prey model with a polymorphic population similar to our methods shows examples of RQ dynamics.

Collecting evidence of RQ dynamics from natural systems is a difficult task \cite{abrams2000evolution} and in experiments the long term of the measurements is a problem. One experiment that shows the effect of evolution on predator-prey dynamics \cite{mougi2010evolution} suggests RQ dynamics, but due to the measurement time it cannot be excluded that the cycles will dampen. While there is an increase in long term adaptation experiments these are usually in bacterial systems while most of the modelling has focused on sexually reproducing predator-prey communities. Here we have tried to include these systems by allowing for phenotypic variation and larger mutation effect sizes. In bacterial systems polymorphic populations are common, even under reasonably constant conditions as recently shown in the Long Term Evolution Experiment with \textit{E. coli} \cite{good2017dynamics}. The polymorphic trait distribution might result from density dependent dynamics, as shown theoretically in a chemostat \cite{wortel2016evolutionary}. Therefore our results are an important addition for linking theoretical to experimental observations of evolutionary dynamics.

Although we have tried to keep the modelling in this paper general, our results have some limitations. We show that with some assumptions we can guarantee that continual dynamics will arise with enough time scale separation. Systems that do not follow our assumptions (e.g. fast negative feedbacks, no positive feedback) are more likely to lead to stasis. However we do not prove any conditions that necessarily lead to stasis. Although we tried to keep our equations general, we did not include individual variation within the population. It would be interesting to see if we can combine the methods used in this paper with the polymorphic population and link of the ecological and evolutionary timescales with individual-based models. We limited our results here to only one evolving population, but the feedback we describe could come from another evolving species or be a result of a complex ecosystem. A next step would be to extend this mathematical approach to more than one evolving species. The results in this paper are constrained to micro-evolutionary dynamics, it would be interesting to extend these results to macro-evolutionary phenomena. In this line, there is a nice analogy with the results of \cite{doebeli2017diversity}. They show that most coevolutionary dynamics are found with intermediate diversity, where not all niches are filled. Here we see that if we get diversification, and therefore more different phenotypes, the co-evolutionary dynamics cease.

Without constraining both the functional forms used in our model and the phenotypic diversity within the system we have demonstrated that a fast positive feedback combined with a slow negative feedback always leads to continual dynamics with the proper timescales. By so doing we have improved the understanding of continual evolution and co-evolution in a large class of models, and may be used to predict evolutionary dynamics without knowledge of the exact equations describing the system.

\newpage

\section{Supplementary information to:\\Coupled fast and slow feedbacks lead to continual evolution: A general modeling approach}

\subsection{One species with two phenotypes and a fixed population size} \label{sec:twophenofixedpopsize}

\subsubsection{The system}

We first consider the simplest possible system where the combination of a slow negative and a fast but weaker positive feedback cause continual evolution. We consider a single species with two strains $A$ and $B$, and we assume that the total population size $A+B$ remains constant. We will write
$$
R = \frac{A}{A+B},
$$
hence $R$ takes on values in the interval $[0,1]$. We introduce a second variable $\varphi$ playing the role of a negative feedback, and will be assumed to also take on values in the interval $[0,1]$. We assume that the rates at which the concentrations $A$, $B$ and $\varphi$ vary depends only on these three variables. Since $A+B$ is assumed to remain constant we can reduce the system of differential equations to the following suggestive form
$$
\begin{aligned}
\frac{dR}{dt} & = f(R,\varphi) + \epsilon_m \cdot g(R),\\
\frac{d\varphi}{dt} & = \epsilon_e \cdot h(R,\varphi).
\end{aligned}
$$
Here $f, g$, and $h$ are assumed to be continuously differentiable, and the constants $\epsilon_m$ and $\epsilon_e$ are assumed to be sufficiently small. The function $f$ represents the ecological dynamics of the populations $A$ and $B$. Since an extinct population cannot reproduce we naturally obtain the extinction assumption
\begin{enumerate}
\item[(EXT)] The function $f$ takes on the value $0$ when $R \in \{0,1\}$.
\end{enumerate}

The function $g$ represents mutations from population $A$ to population $B$ and vice versa. While $\varphi$ will influence the ecology, we assume that $\varphi$ does not affect the mutation rates. Since it is natural to assume that the number of mutations increases with the size of the population, we obtain the mutation assumption
\begin{enumerate}
\item[(MUT)] The function $g$ is decreasing in $R$, $g(0) > 0$ and $g(1) < 1$.
\end{enumerate}

The constants $\epsilon_m$ and $\epsilon_e$ represent the different time scales of the dynamical system. By choosing $\epsilon_m$ and $\epsilon_e$ sufficiently small, we guaranteed that slower time scales for the mutations and for the ecology of $\varphi$ than for the ecology of $A$ and $B$. 

We make three further assumptions: a condition on $f$ that guarantees a positive feedback in the $R$-variable, and two conditions that combined guaranteed a stronger negative feedback on $R$ caused by the slowly adapting variable $\varphi$.

\begin{enumerate}
\item[(PF)] The function $f$ is strictly decreasing in $\varphi$, and
$$
\frac{\partial f}{\partial R} > 0.
$$
\item[(NF1)] The function $h$ is decreasing in $\varphi$ and increasing in $R$.
\item[(NF2)] There exist $0 < \varphi_{--} < \varphi_{++} < 1$ such that for $\varphi < \varphi_{--}$, the value of $f$ is strictly positive for all $R\in(0,1)$, and for $\varphi > \varphi_{++}$, the value of $f$ is strictly negative for all $R\in(0,1)$. Similarly, given any $\varphi \in (0,1)$ the function $h(R,\varphi)$ is stricly positive for $R$ suficiently large, and strictly negative for $R$ sufficiently small.
\end{enumerate}

We will prove that this combination of assumptions guarantees large fluctuations in the values of $R$ and $\varphi$.

Let $K \subset (0, 1) \times (0,1)$ be a closed subset of states $(R, \varphi)$.

\begin{lemma}\label{lemma1}
When $\epsilon_m$ and $\epsilon_e$ are chosen sufficiently small, the orbit of almost every initial condition $(R_0, \varphi_0) \in K$ will eventually leave $K$.
\end{lemma}
\begin{proof}
By assumption (iii), the differential equation will have strictly positive divergence on $K$ by choosing $\epsilon_m$ and $\epsilon_e$ sufficiently small. In other words, the dynamical system will be strictly area expanding on $K \cap \{R \ge 1\}$ with respect to the standard Euclidean area. It follows that it is not possible for the orbits of a set of positive area to remain in $K$.
\end{proof}

Note that for this result to hold on a given set $K$, assumptions (iii) is only necessary on $K$. In fact, in our simulations below these two assumptions will not be globally satisfied.

\begin{lemma}
By choosing $\epsilon_m$ and $\epsilon_e$ sufficiently small, almost every orbit converges to an invariant closed curve not completely contained in $K$.
\end{lemma}
\begin{proof}
By the Poincar\'e-Bendixson Theorem, every orbit must converge to either a fixed point or a invariant closed curve.

Independently of the values of $\epsilon_m$ and $\epsilon_e$, we can construct a large rectangle $[a,b] \times [c,d] \subset (0,1) \times (0,1)$ that must contain all fixed points. Indeed, by choosing $c < \varphi_{--}$ and $d > \varphi_++$ it follows that there are no fixed points for $\varphi<c$ or $\varphi>d$. By assumption (NF2) we can choose $a$ sufficiently small and $b$ sufficiently large so that for $\varphi \in [c,d]$ and $R<a$ the variable $\varphi$ must decrease, while for $\varphi \in [c,d]$ and $R>b$ the variable $\varphi$ must increase. Thus, there are no fixed points outside of the rectangle $[a,b] \times [c,d]$.

As remarked above, the dynamical system can be made area expanding in the rectangle by choosing $\epsilon_m$ and $\epsilon_e$ sufficiently small. Thus, the fixed points in $[a,b] \times [c,d]$ cannot be stable, nor can the rectangle contain an invariant closed curve.
\end{proof}

\begin{thm}
For $\epsilon_e$ and $\epsilon_m$ sufficiently small almost every orbit will converge to a periodic orbit that intersects both $R> b$ and $R<a$.
\end{thm}
\begin{proof}
Let $0<R_- < R_+< 1$ be such that $h$ is strictly positive for $R > R_+$ and $\varphi < \varphi_{++}$, and strictly negative for $R<R_-$ and $\varphi> \varphi_{--}$. Since $f$ is strictly increasing in $R$ and strictly decreasing in $\varphi$, the curve $\{f = 0\}$ is a strictly increasing graph that does not intersect $\{\varphi >\varphi_{++}\}$ or $\{\varphi< \varphi_{--}\}$. Denote its graph by $\Gamma$, and write $(R_+, \varphi_+)$ for the intersection point of $\Gamma$ with $\{R = R_+\}$. Let $R_{++} > R_+$.

By choosing $\epsilon_m$ sufficiently small we can guarantee that $\frac{dR}{dt}>0$ for $R = R_{++}$ and $\varphi < \hat{\varphi}_+$, where $\hat{\varphi}_+ \in (\varphi_+, \varphi_{++})$. It follows that whenever an orbit intersects the rectangle $[R_{++}, 1] \times [0, \hat{\varphi}_+]$, the orbit will remain in this region for finite time until it crosses the upper boundary in some point $(R_t, \hat{\varphi}_+)$. Note that mutations prevent either population from going extinct, so $R_t< 1$.

At time $t$, the variable $\varphi$ will continue to increase for as long as $\varphi< \varphi_{++}$ and $R \ge R_+$. But since $f$ is negative for $\varphi \ge \varphi_{++}$, it follows that the orbit must cross the line $\{R = r_+\}$. The first time this occurs must therefore be in a point $(R_+, \varphi_{t^\prime})$ for which $\varphi_{t^\prime} > \hat{\varphi}_+$.

Since $\hat{\varphi}_+$ is strictly larger than $\varphi_+$, it follows that $f<0$ for all $R \le R_+$ and all $\varphi > \varphi_+^\prime$, where $\varphi_+^\prime \in (\varphi_+, \hat{\varphi}_+)$ can be chosen independent of the exact intersection point $(R_+, \hat{\varphi}_+)$. By choosing $\epsilon_e$ and $\epsilon_m$ sufficiently small, the orbit can be guaranteed to reach any sublevel $R< R_{--} < R_-$ while $\varphi$ remains at least $\varphi_+^\prime$. Since our assumptions are completely symmetric, a symmetric argument shows that $R$ remains smaller than $R_{--}$ until $\varphi < \varphi_-$, where $(R_-, \varphi_-) \in \Gamma$, and later returns to the region $R> R_{++}$.

As remarked earlier, we can make sure that all fixed points are contained in the rectangle $K = [R_{--}, R_{++}], [\varphi_-^\prime, \varphi_+^\prime]$, where $\varphi_-^\prime< \varphi_-$ is defined analogously as $\varphi_+^\prime$. Since $\varphi_+^\prime>\varphi_+$ and $\varphi_-^\prime<\varphi_-$ can be chosen independently of the points $(R_+, \hat{\varphi}_+)$ and $(R_-, d_-)$, it follows that the entire forward orbit avoids the rectangle $K$. Thus the Carleson-Bendixson Theorem implies that the orbit must converge to a periodic cycle. 
By Lemma \ref{lemma1} we can guarantee that almost every orbit in $K$ must eventually leave $K$, which completes the proof.
\end{proof}

\begin{figure}
\includegraphics[width=4in]{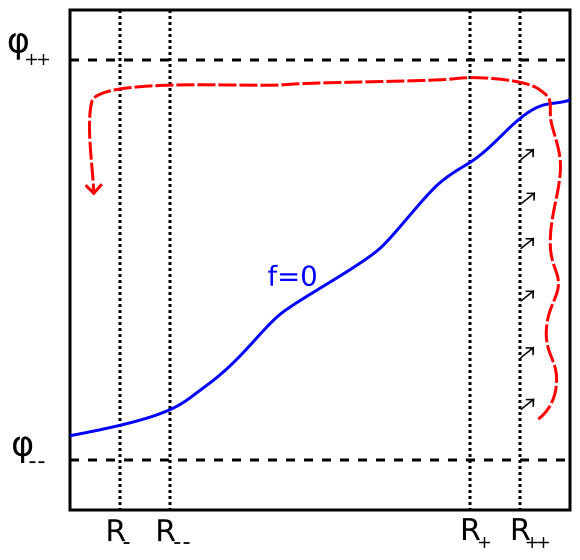}
\caption{In red: orbit avoiding a large rectangle.}
\end{figure}

\subsection{One species with two phenotypes and a variable population size} \label{sec:twophenovariablepopsize}

In the previous section we assumed that $\frac{dR}{dt}$ depends only on $R$ and $\varphi$. When the total population size is not assumed to be fixed, such an assumption is not realistic. However, it turns out that the mechanism causing the drastic fluctuations can still occur without this assumption. The important observation is that to prove the existence of large fluctuations, we merely need to know the \emph{signs} of $\frac{dR}{dt}$ and $\frac{d\varphi}{dt}$, and not the exact values.

Suppose now that we have two distinct traits, with population sizes $A$ and $B$. We do not assume that $A+B$ is constant, but will still write
$$
R = \frac{A}{A+B}.
$$
Thus $R$ takes on values in the interval $[0,1]$. We assume that the total population will remain bounded, from above as well as below, so that $R \rightarrow 0$ corresponds to the extinction of population $A$, and $R \rightarrow 1$ to the extinction of population $B$. One naturally assumes that apart from mutations $dA/dt=0$ when $A = 0$, and similarly $dB/dt = 0$ when $B = 0$, to obtain a system of differential equations of the form
$$
\begin{aligned}
\frac{dA}{dt} & = A \cdot f_A(A, B, \varphi) + \epsilon_m \cdot g(A,B),\\
\frac{dB}{dt} & = B \cdot f_B(A,B, \varphi) - \epsilon_m \cdot g(A,B), \; \; \mathrm{and}\\
\frac{d\varphi}{dt} & = \epsilon_e \cdot h(A, B, \varphi).
\end{aligned}
$$
As before, $\varphi$ will play the role of a strong negative feedback. We will assume that the functions $f_A, f_B, g,$ and $h$ are all continuously differentiable, and the constants $\epsilon_m$ and $\epsilon_e$ will be chosen arbitrarily small.

The existence of a positive ecological feedback on the populations $A$ and $B$ leads to the following assumption:
\begin{enumerate}
\item[(PF)] For all values of $A, B,$ and $\varphi$ we have
$$
\left(\frac{\partial}{\partial A} - \frac{\partial}{\partial B}\right) f_A  > 0, \; \; \mathrm{and} \; \;
\left(\frac{\partial}{\partial B} - \frac{\partial}{\partial A}\right) f_B  > 0.
$$
\end{enumerate}

Thus, if an amount of $B$ is replaced by an equal amount of $A$ while $\varphi$ remains fixed, the fitness of the population $A$ increases while the fitness of the population $B$ decreases.

The next assumption, which we will refer to as the \emph{unique stable value assumption}, will be used to draw conclusions about the sign of $\frac{dR}{dt}$ without knowing the exact values of $A$ and $B$. We do not claim that this assumption is necessary, but it turns out to be satisfied in many models and is convenient:
\begin{enumerate}
\item[(USV1)] For fixed values of $B$ and $\varphi$, there is a unique value $A_0 \ge 0$ such that
$$
f_A(A, B, \varphi) > 0
$$
for $0 \le A < A_0$, and
$$
f_A(A,B, \varphi) < 0
$$
for $A > A_0$.
\item[(USV2)] For fixed values of $A$ and $\varphi$, there is a unique value $B_0 \ge 0$ such that
$$
f_B(A, B, \varphi) > 0
$$
for $B < B_0$, and
$$
f_B(A,B, \varphi) < 0
$$
for $B > B_0$.
\end{enumerate}

For $\varphi$ fixed we can therefore consider $A_0$ as a graph $\eta_A(B, \varphi)$, and similarly $B_0 = \eta_B(A, \varphi)$. The positive feedback condition guarantees the following:
\begin{lemma}\label{lemma:graph}
The derivatives $\frac{\partial \eta_A}{\partial B}$ and $\frac{\partial \eta_B}{\partial A}$ are both strictly less than $-1$, and possibly $-\infty$.
\end{lemma}
\begin{proof}
We prove the first statement, the second is analogous. Let $A_0, B, \varphi$ be a triple for which $f_A(A_0, B, \varphi) = 0$. The unique stable value assumption implies that
$$
\frac{\partial f_A}{\partial A}(A_0, B, \varphi) \le 0,
$$
and the positive feedback assumption gives
$$
\frac{\partial f_A}{\partial B} < \frac{\partial f_A}{\partial A}.
$$
The function $\eta_A$ is implicitly defined by $f_A(\eta_A(B, \varphi), B, \varphi) =0$. When $\frac{\partial f_A}{\partial A} \neq 0$ the Implicit Function Theorem gives
$$
\frac{\partial \eta_A}{\partial B} = -\frac{\frac{\partial f_A}{\partial B}}{\frac{\partial f_A}{\partial A}}
$$
which gives $\frac{\partial \eta_A}{\partial B} < -1$. When $\frac{\partial f_A}{\partial A} = 0$ we have $
\frac{\partial f_A}{\partial B} < 0$ by assumption (PF), and hence $\frac{\partial \eta_A}{\partial B} = -\infty$.
\end{proof}

We assume the existence of a strong negative feedback:
\begin{enumerate}
\item[(NF1)] The difference $f_A - f_B$ decreases with $\varphi$, while
$$
\left(\frac{\partial}{\partial A} - \frac{\partial}{\partial B}\right) h > 0.
$$
\item[(NF2)] There exists $0< \varphi_{--} < \varphi_{++} < 1$ such that
$$
f_B > f_A
$$
whenever $\varphi> \varphi_{++}$, while
$$
f_A > f_B
$$
whenever $\varphi < \varphi_{--}$.
\end{enumerate}

Hence, for $\varphi$ sufficiently large and $\varphi$ sufficiently small the level sets $\{f_A = 0\}$ and $\{f_B = 0\}$ do not intersect. We may write $\varphi_{++}$ and $\varphi_{--}$ for the maximal resp. minimal value of $\varphi$ for which the level sets intersect. For our argument it will be necessary that both $\varphi_{++}$ and $\varphi_{--}$ can be reached, hence we make a ``transitivity'' assumption:

\begin{enumerate}
\item[(Tr)] The function $h$ is strictly positive when $\varphi \le \varphi_{++}$, $B=0$ and $A = \eta_A(0, \varphi_{++})$. Similarly, the function $h$ is strictly negative when $\varphi \ge \varphi_{--}$, $A=0$ and $B = \eta_B(0, \varphi_{--})$.
\end{enumerate}

Let us consider, for a fixed value of $\varphi$, joint solutions of the two equations
$$
A\cdot f_A = 0 \; \; \mathrm{and} \;\; B \cdot f_B = 0.
$$
It is clear that, besides the origin, there always is a unique solution on each of the axes. For $A, B > 0$ it follows from Lemma \ref{lemma:graph} that there either is no solution or a unique solution, depending on the value of $\varphi$. For $\varphi = \varphi_{++}$ the intersection point of $\{f_A = 0\}$ and $\{f_B = 0\}$ lies in the axis $\{A = 0\}$, while for $\varphi = \varphi_{--}$ the intersection point lies in the axis $\{B = 0\}$. See Figure \ref{figure:phase1} for a simple depiction of the level sets $\{f_A = 0\}$ (in red) and $\{f_B = 0\}$ (in blue).

\begin{figure}
\label{figure:phase1}
\includegraphics[width=2in]{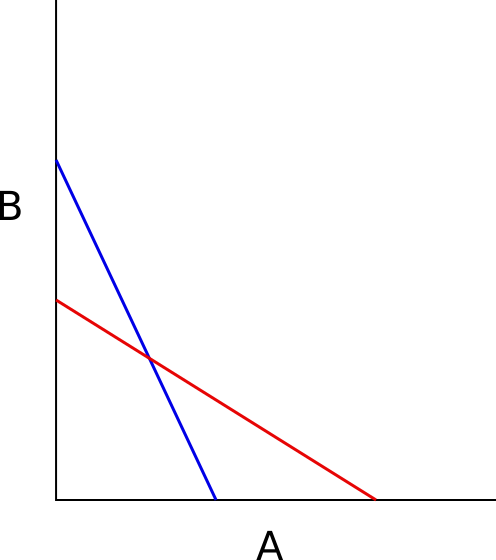}
\caption{The graphs of $\eta_A$ and $\eta_B$ for $\varphi_{--} < \varphi < \varphi_{++}$.}
\end{figure}

\medskip

Let us consider the ecological dynamics in the $(A,B)$-plane caused by the differential equations
$$
\begin{aligned}
\frac{dA}{dt} & = A \cdot f_A,\\
\frac{dB}{dt} & = B \cdot f_B
\end{aligned}
$$
for a fixed value of $\varphi$. When $\varphi \le \varphi_{--}$ or $\varphi \ge \varphi_{++}$ there are three fixed points where
$$
\frac{dA}{dt} = \frac{dB}{dt} = 0,
$$
namely the origin and the two intersection points of the curves $\{f_A = 0\}$ and $\{f_B = 0\}$ with the respective axes $\{B = 0\}$ and $\{A = 0\}$.
The origin is always repelling. One of the points on the axes is a saddle fixed point, with stable manifold equal to the axis.
The third fixed point is attracting, and all orbits of initial values not lying on the axes converge to this attracting fixed point.

When $\varphi_{--} < \varphi < \varphi_{++}$ there are four fixed points. Again the origin is a repelling fixed point. There are again two fixed points on the axes, which are now both attracting. Finally, there is an intersection point of the curves $\{f_A = 0\}$ and $\{f_B = 0\}$, and assumption (i) implies that this is a hyperbolic saddle fixed point. Its stable manifold is the separatrix of the two attracting basins.

\begin{figure}
\label{figure:phase3}
\includegraphics[width=4in]{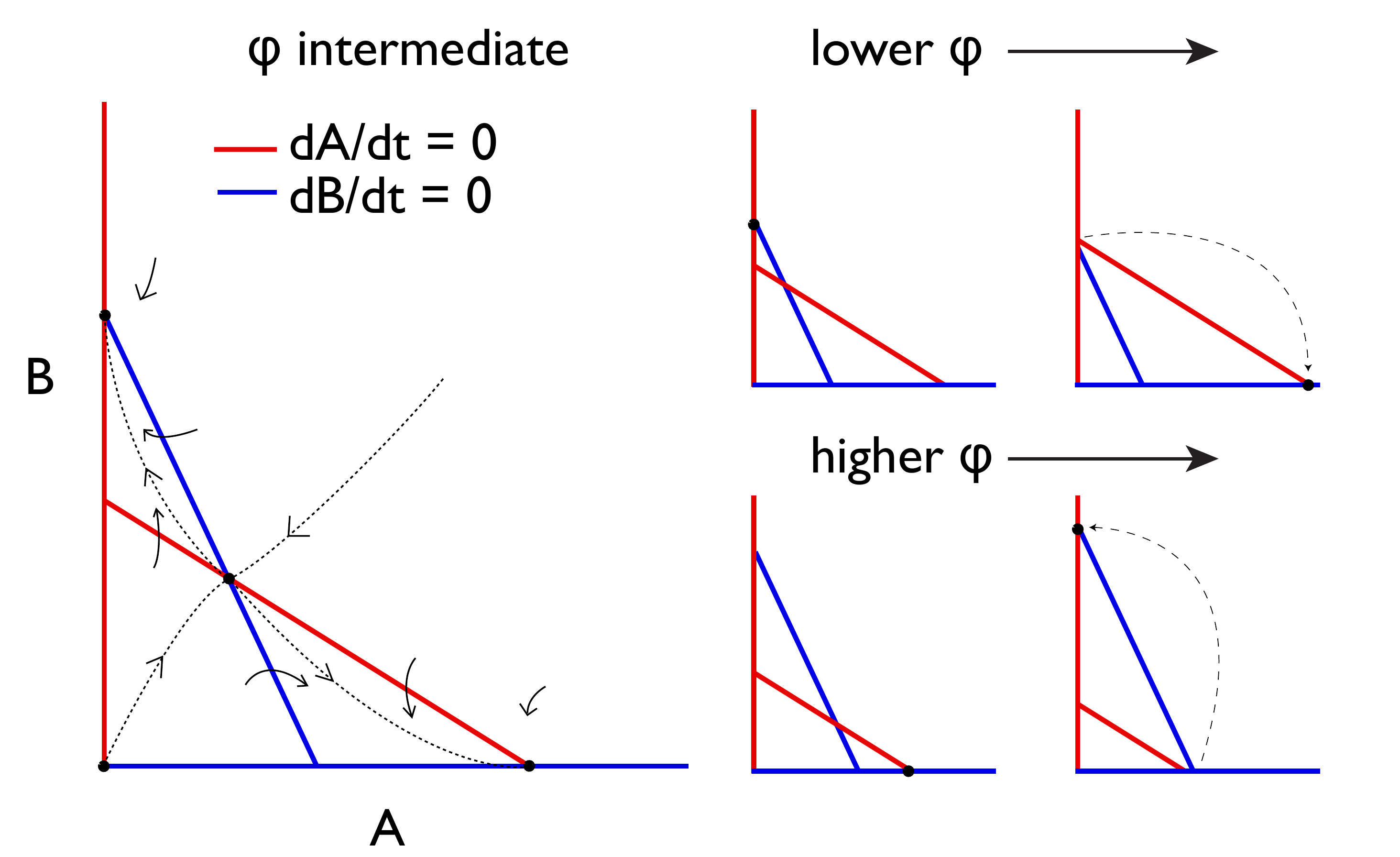}
\caption{Left: Dynamics in the $(A,B)$-plane for fixed $\varphi_{--} < \varphi < \varphi_{++}$. Right: Changes in the phase plane when $\varphi$ decreases below $\varphi_{--}$ (top) or increases above $\varphi_{++}$ (top).}
\end{figure}

\medskip

Let us now consider the effect of the mutations, represented by $\epsilon_m \cdot g(A,B)$, on the dynamics in the $(A,B)$-plane in the case $\varphi_{--} < \varphi < \varphi_{++}$. The behavior near each of the fixed points is stable under small $C^1$-perturbations, and the qualitative behavior of the system is robust. Hence, by choosing $\epsilon_m$ sufficiently small, there will still be a repelling fixed point at the origin, and a saddle point with separatrix near the intersection point of the curves $\{f_A = 0\}$ and $\{f_B = 0\}$. The rest of the quadrant is attracted to neighborhoods of the original attracting fixed points. These neighborhoods can be chosen arbitrarily small by choosing $\epsilon_m$ sufficiently small. Note that the addition of mutations causes the axes to be repelling, hence the attracting fixed point no longer lies on the axis but sufficiently nearby.

\medskip

Finally let us consider the full three-dimensional dynamical system, taking into account that $\varphi$ is not fixed. We assume that for given $t_0$ we have $\varphi(t_0) < \varphi_{++}$, and that $(A(t_0),B(t_0))$ lies in the small attracting region near the $A$-axis. For any $\delta>0$ we can, by taking $\epsilon_e$ and $\epsilon_m$ sufficiently small, assume that $A(t_0) - \eta_A(0,\varphi_{++}) > - \delta$ and $B < \delta$. It follows that $h(A(t_0), B(t_0), \varphi(t_0))$ is strictly positive. It follows that $\varphi$ will continue to grow while $\varphi < \varphi_{++} + \delta$ and while $(A(t),B(t))$ remains trapped in the small attracting neighborhood near the $A$-axis.

Let us consider the first time $t_1$ at which either $\varphi \ge \varphi_{++} + \delta$ or at which $(A, B)$ leaves the small attracting neighborhood. In either case it follows that the orbit $(A(t),B(t))$ is guaranteed to approach the small attracting neighborhood near the $B$-axis. If $\epsilon_e$ is sufficiently small the value of $\varphi$ can only decrease arbitrarily little while this happens. The conclusion is that we end up with a time $t_2>t_1$ when $\varphi(t_1) > \varphi_{--}$ and $(A(t_2), B(t_2))$ lies in the attracting neighborhood near the $B$-axis. By the symmetry of our assumptions the process will repeat itself. We have proved the following.

\begin{thm}
Let $0<a<b<1$, $\varphi_{--} < c < d< \varphi_{++}$, and write $K$ for the rectangle $[a,b]\times [c,d]$. Then for $\epsilon_e$ and $\epsilon_m$ sufficiently small there exists an orbit $(A(t), B(t), \varphi(t)$ for which the coordinates $(R(t), \varphi(t))$ avoid $K$, and for which $R(t)$ fluctuates between values larger than $b$ and smaller than $a$.
\end{thm}

\subsection{One species with a trait with multiple phenotypes} \label{sec:multivaluedtrait}

Let us recall that in the first setting, where we considered two traits and a constant total population size, the argument used that the rate at which the ratio $R = \frac{A}{A+B}$ changes is determined only by $R$ and $\varphi$. When the assumption on the constant population size was dropped, this was no longer the case, but we saw that the main idea of the argument can still be used because in the appropriate regions in $(R,\varphi$)-coordinates the \emph{sign} of $\frac{\partial R}{\partial t}$ is still known. One can easily imagine that, under appropriate assumptions, the same approach can be used when dealing with more than two phenotypes. The good choice of $R$, a projection from many variables onto $\mathbb R$, is of course essential to the argument.

Let us present a particular family of systems with finitely many trait values. We consider $n+1$ different strains $u_i$, where the index $i \in [0,1]$ is of the form $\frac{j}{n}$. As before, we will stipulate a fast positive feedback and a slow but dominating negative feedback caused by a variable $\varphi$. We will assume that $\varphi$ is produced more by strains $u_i$ for $i$ large, and for fixed values of $u  = (u_0, \ldots , u_1)$ converges towards the average index $M$ given by
$$
M := \frac{\sum_{i = 0}^{1} i u_i}{\sum_{i = 0}^1 u_i}.
$$
In the setting with multiple traits, knowledge of $M$ and $\varphi$ will not be sufficient to determine $dM/dt$ and $d\varphi/dt$. However, we will present assumptions where the \emph{signs} of $dM/dt$ and $d\varphi/dt$ can be determined, at least in the relevant regions, which will be sufficient to guarantee large fluctuations.

Looking first only at the ecological factors of the dynamical system, ignoring the effects due to mutations, we assume that the densities of the individuals with phenotype $i$ change according to the equations
$$
\frac{du_i}{dt} = u_i \left(\mu_i \left(1 - \frac{\sum_{i = 0}^{1} u_i}{k}\right) - d\right),
$$
where the growth rates $\mu = (\mu_0, \ldots , \mu_1)$ are strictly positive, depend continuously on $u$ and $\varphi$, and for any given values of $u$ and $\varphi$ the tuple $(\mu_0, \cdots, \mu_1)$ is assumed to be either independent of $i$ or strictly monotonic in $i$. 
When $\mu$ is constant the relative population sizes do not change, and it follows that $\frac{dM}{dt} = 0$.

\begin{lemma}\label{lemma:agree}
When $\mu$ is strictly increasing (resp. decreasing) in $i$, the rate $\frac{dM}{dt}$ is strictly positive (resp. negative).
\end{lemma}
\begin{proof}
Let us assume that $\mu$ is strictly increasing, the argument is identical when $\mu$ is decreasing. We note that
$$
\begin{aligned}
\frac{dM}{dt} & = \frac{d}{dt} \left(\frac{\sum_{i = 0}^{1} i u_i}{\sum_{i = 0}^1 u_i}\right)\\
 & = \frac{\left(\sum_{i = 0}^{1} i \frac{du_i}{dt}\right)\left(\sum_{i = 0}^1 u_i \right) - \left(\sum_{i = 0}^{1} i u_i\right)\left(\sum_{i = 0}^1  \frac{du_i}{dt} \right)}{\left(\sum_{i = 0}^1 u_i\right)^2}.
\end{aligned}
$$
Since we are only interested in the sign of $\frac{dM}{dt}$, we can drop the denominator and hence need to prove that
$$
\frac{\sum_{i = 0}^{1} i \frac{du_i}{dt}}{\sum_{i = 0}^{1} i u_i} > \frac{\sum_{i = 0}^{1} \frac{du_i}{dt}}{\sum_{i = 0}^{1} u_i}.
$$
Plugging in the formula for $\frac{du_i}{dt}$, we note that the terms $(1 - \frac{\sum_{i = 0}^{1} u_i}{k})$ and $-d$, which are both independent of $i$, drop out of the quotients and we are left with showing that
$$
\frac{\sum_{i = 0}^{1} i \mu_i \cdot u_i}{\sum_{i = 0}^{1} i u_i} > \frac{\sum_{i = 0}^{1} \mu_i \cdot u_i}{\sum_{i = 0}^{1} u_i},
$$
which is equivalent to
$$
\frac{\sum_{i = 0}^{1} i \mu_i \cdot u_i}{\sum_{i = 0}^{1} \mu_i \cdot u_i} > \frac{\sum_{i = 0}^{1} i u_i}{\sum_{i = 0}^{1} u_i}.
$$
Since $\mu$ is assumed to be increasing in $i$ this inequality holds.
\end{proof}

We will assume that the rate of change of the feedback variable $\varphi$ can be written as
$$
\frac{d\varphi}{dt} = \epsilon_e \cdot h(u, \varphi),
$$
where the constant $\epsilon_e$ will later be assumed to be sufficiently small and $h$ is continuously differentiable. We make the following ``negative feedback'' assumption:

\begin{enumerate}
\item[(NF1)] The function $h$ satisfies $h<0$ when $\varphi > M$, and $h > 0$ when $\varphi < M$.
\end{enumerate}

An example of such a function is $h(u) = M(u) - \varphi$, but it is not necessarily the case that $h$ can be expressed as a function of $M$. The NF1 assumption of course implies that for each value of $M$ there is a unique value $\varphi = \varphi(M)$ for which $h = 0$. The assumption that this unique value equals $M$ is merely a convenience, which by the Implicit Function Theorem can be obtained by a change of coordinates whenever $\frac{\partial h}{\partial \varphi}(u, \varphi) \neq 0$ is satisfied for all $u$ and $\varphi = \varphi(M)$.

We will make the following additional assumptions guaranteeing both a negative and a positive feedback.
\begin{enumerate}
\item[(NF2)] There exist $1> \varphi_+ > \varphi_- > 0$ with the following property: For any given non-zero $u$, the rates $\mu$ are decreasing in $i$ for $\varphi > \varphi_+$, and increasing in $i$ for $\varphi < \varphi_-$.
\item[(PF)] For each $\varphi$ there is a unique value $M_\varphi$ such that $\mu$ is strictly increasing at $(u, \varphi)$ whenever $M(u) > M_\varphi$, and strictly decreasing whenever $M(u) < M_\varphi$. We assume that $M_\varphi$ is non-decreasing with $\varphi$.
\end{enumerate}

In other words, when $\mu$ is non-decreasing, it must remain so when $\varphi$ is decreased or when $M(u)$ is increased, and in the latter case must become strictly increasing. Similarly, when $\mu$ is non-increasing, it must remain so when $\varphi$ is increased or when $M(u)$ is decreased, in the latter case it must again become strictly decreasing.

By continuity of $\mu$ it follows from (PF) that $\mu$ is constant when $M(u) = M_\varphi$. Note that we may redefine $\varphi_+$ as the smallest $\varphi$ for which $M_\varphi = 1$, and similarly $\varphi_-$ as the largest $\varphi$ for which $M_\varphi = 0$. Note also that we do not assume that $\mu$ is a function of $M$.


Let us now add mutations to the model:
$$
\frac{du_i}{dt} = u_i \left(\mu_i \left(1 - \frac{\sum_{i = 0}^{1} u_i}{k}\right) - d\right) + \epsilon_m\cdot g_i(u),
$$
We assume that each function $g_i$ is non-negative when $u_i = 0$, is decreasing in $u_i$, and is strictly increasing in each $u_j$ for $j \neq i$. In particular $g_i$ is strictly positive when $u_i = 0$ but $u \neq 0$, thus mutations from any strain to any other strain are possible. We will later discuss alternative assumptions, making it possible to restrict some of the mutations.

\begin{lemma}\label{lemma:order}
Each $u_i$ remains bounded from below by a constant of order $O(\epsilon_m)$.
\end{lemma}
\begin{proof}
It follows from the formula for $\frac{du_i}{dt}$ that $\|u\|:=\sum_{i = 0}^{1} u_i$ remains bounded from above and below, i.e. $\|u\| = O(1)$. It follows that as $u_i \rightarrow 0$:
$$
\epsilon_m \cdot g_i(u) \ge O(\epsilon_m).
$$
Thus for $u_i$ small, the worst case scenario is that
$$
\frac{du_i}{dt} \ge - C \cdot u_i + c\cdot \epsilon_m,
$$
for some uniform constants $c, C>0$. It follows that
$$
u_i \ge \frac{c \epsilon_m}{C} = O(\epsilon_m).
$$
\end{proof}

\begin{thm}\label{thm:zeven}
For $\epsilon_e$ and $\epsilon_m$ sufficiently small there exist orbits $(u(t), \varphi(t)$) for which $M(t)$ fluctuates arbitrarily closely between $0$ and $1$.
\end{thm}
\begin{proof}
Suppose that we start with initial values $u(t_0), \varphi(t_0)$ for which $M(t_0) := M(u(t_0)) > \varphi_+$, and for which $\mu$ is increasing in $i$. It follows that $\varphi(t_0) < \varphi_+$, hence $\varphi(t_0) < M(t_0)$ and therefore $\varphi$ is increasing at time $t_0$.

By choosing $\epsilon_m$ sufficiently small we can guarantee that $M(t)$ remains arbitrarily close to $1$ until $\varphi$ is arbitrarily close to $\varphi_+$, say $\varphi > \varphi_+ - \delta$ for $\delta > 0$ arbitrarily small. Note that $\varphi$ remains increasing while this is the situation, hence at some time $t_1>t_0$ we must have $M(t_1)= M_{\varphi(t_1)}$. We may assume that $t_1$ is the first time that this occurs, from which it follows that $M(t_1)$ must still be arbitrarily close to $1$, while $\varphi(t_1)$ must be arbitrarily close to $\varphi_+$. In particular $\varphi$ is still increasing at time $t_1$, while $M(t)$ is decreasing due to mutations. As a consequence (PF) implies that $\mu$ becomes strictly decreasing, hence $M(t)$ will continue to decrease.

Since $\varphi$ remains increasing and $M(t)$ remains decreasing, it follows that there is a smallest time $t_2 > t_1$ for which $M(t_2) = \varphi(t_2)$. We may assume that $t_2$ is the first time at which equality occurs, from which it follows that $\varphi$ has only increased between $t_1$ and $t_2$, and hence $\varphi(t_2) > \varphi_+ - \delta$.

We claim that $u_0(t_2)$ is bounded from below by a constant that is independent from $\epsilon_m$. To see this, note that by the assumption that the vector $u$ remains bounded, it follows that the $u_i$'s for $i \neq 0$ must remain arbitrarily small for $t \in (t_0, t_1)$. It follows that the corresponding growth factors $\mu_i \left(1 - \frac{\sum_{i = 0}^{1} u_i}{k}\right) - d$ must remain strictly negative, with a uniform bound from above. It follows that populations $u_i$ for $i \neq 1$ remain comparable to $\epsilon_m$, where Lemma \ref{lemma:order} implies the estimate from below, and in particular the populations $u_i(t_1)$ for $i \neq 1$ are comparable to each other, with ratios independent of $\epsilon_m$. Recall that for $t \in [t_1,t_2]$ we noted that $\mu$ is increasing, and hence the growth factor $\mu_i\left( 1 - \frac{\sum_{i=0}^1 u_i}{k}\right) - d$
is largest for $i = 0$. It follows that in the interval $[t_1, t_2]$ the population $u_0$ grows faster than any other population $u_i$, with a strictly larger exponential coefficient. At time $t_2$ the average $M(t)$ has decreased by an amount independent of $\epsilon_m$. Since the total population remains bounded away from $0$ by assumption, it follows that the size of $u_0$ must have increased by an amount independent of $\epsilon_2$, thus obtaining the claim.

By assuming that $\epsilon_m$ and $\epsilon_e$ are sufficiently small, it follows from (PF) and continuity of $\mu$ that $\mu$ will remain decreasing and $M(t)$ decreases below $\mu_-$, say at time $t_3$, and that the time interval $t_3 - t_2$ is bounded and independent of $\epsilon_m$ or $\epsilon_e$. Since $\epsilon_e$ is assumed to be small, it follows that $\varphi(t_2) \sim \varphi(t_1)$. We have ended up with assumptions on $M(t_3)$ and $\varphi(t_3)$ that are symmetrical to those on $M(t_0)$ and $\varphi(t_0)$. The symmetry of our assumptions implies that the process will repeat itself, causing arbitrarily large fluctuations in $M(u)$.
\end{proof}

If we drop the assumption that mutations from any strain to any other strain are possible, and replace it instead by the much weaker assumption that given any two strains $u_i$ and $u_j$ there is a possible sequence of mutations from $u_i$ to $u_j$, Lemma \ref{lemma:order} fails, and the above proof breaks down. We cannot guarantee that at time $t_2$ the strain $u_0$ has increased to a definite size, independent of $\epsilon_m$, and as a result we cannot give a bound on the time interval $t_3 - t_2$.

This issue can be solved by assuming that the constant $\epsilon_e$ is sufficiently small, where the bound on $\epsilon_e$ may have to depend on the choice of $\epsilon_m$. Note the difference with the above statement, which holds whenever both $\epsilon_m$ and $\epsilon_e$ are sufficiently small. In practice the stronger assumption on $\epsilon_e$, which can imply that $\epsilon_e$ is much smaller than $\epsilon_m$, may or may not be desirable.


\section{Examples: Rate equations, diversity index and supplementary figures}

\subsection{Shannon index}

The Shannon diversity index (H) is used in ecological research to describe species richness. Here we used the measure to describe phenotype richness within a species (we do not discuss whether bimodal phenotype distributions lead to different species). We calculate this diversity with the following formula:

$$
H = - \sum_1^n p_i \ln p_i
$$

Where $n$ is the number of phenotypes and $p_i$ the proportion of phenotype $i$ in the population.

\subsection{One species with two phenotypes and a fixed population size}

In section \ref{sec:twophenofixedpopsize} we give a proof for a system with two phenotypes and a fixed population size, and in this section we made an example of such a system. For this example we used the following equations:

$$
\begin{aligned}
\frac{d R}{d t} & =  R \left( \frac{0.011 + \frac{R}{1+R}}{k_A \left( 1 + \frac{\varphi}{k_B} \right) \left(1+ \frac{0.011+ \frac{R}{1+R}}{k_A} \right)} -0.5 \right) + \epsilon_M \frac{1 - R}{1 + R} \\
\frac{d \varphi}{d t} & = \epsilon_E \left( \frac{R}{1+R} - \varphi \right)
\end{aligned}
$$

$R$ and $\varphi$ are short for $R{t}$ and $\varphi{t}$ since they are time dependent. $R$ is the ratio of the two phenotypes, $A$ and $B$ ($r=\frac{A}{B})$, where $A$ is the phenotype interacting with the negative feedback $\varphi$. Figure \ref{fig:binarySystemR} shows the result of time simulations of the example.

\begin{figure}[!htbp]
  \begin{tabular}{cc}
  \multirow{-3}{*}[3em]{\includegraphics[width=.6 \textwidth]{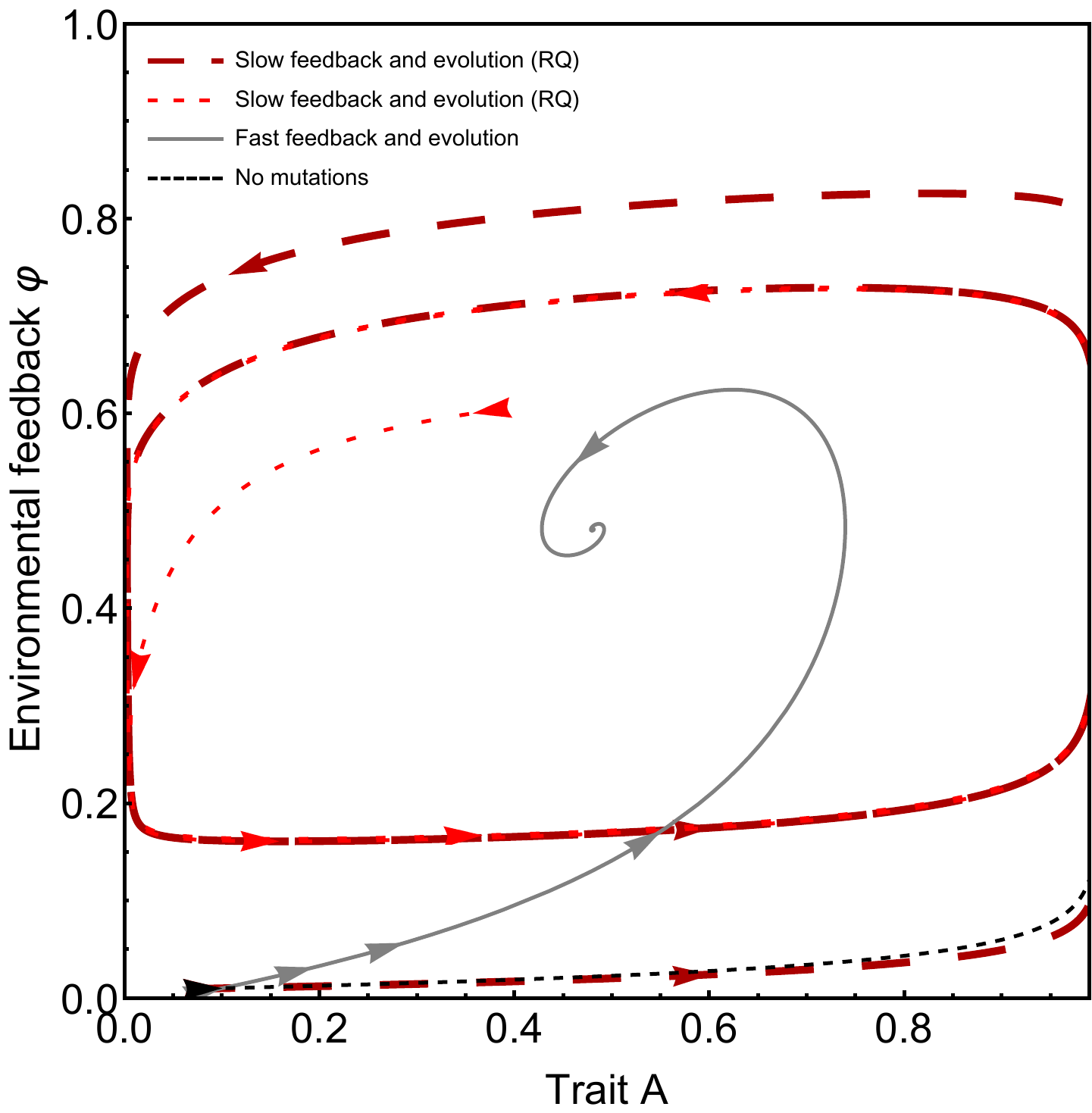}} &
  \includegraphics[width=.4 \textwidth]{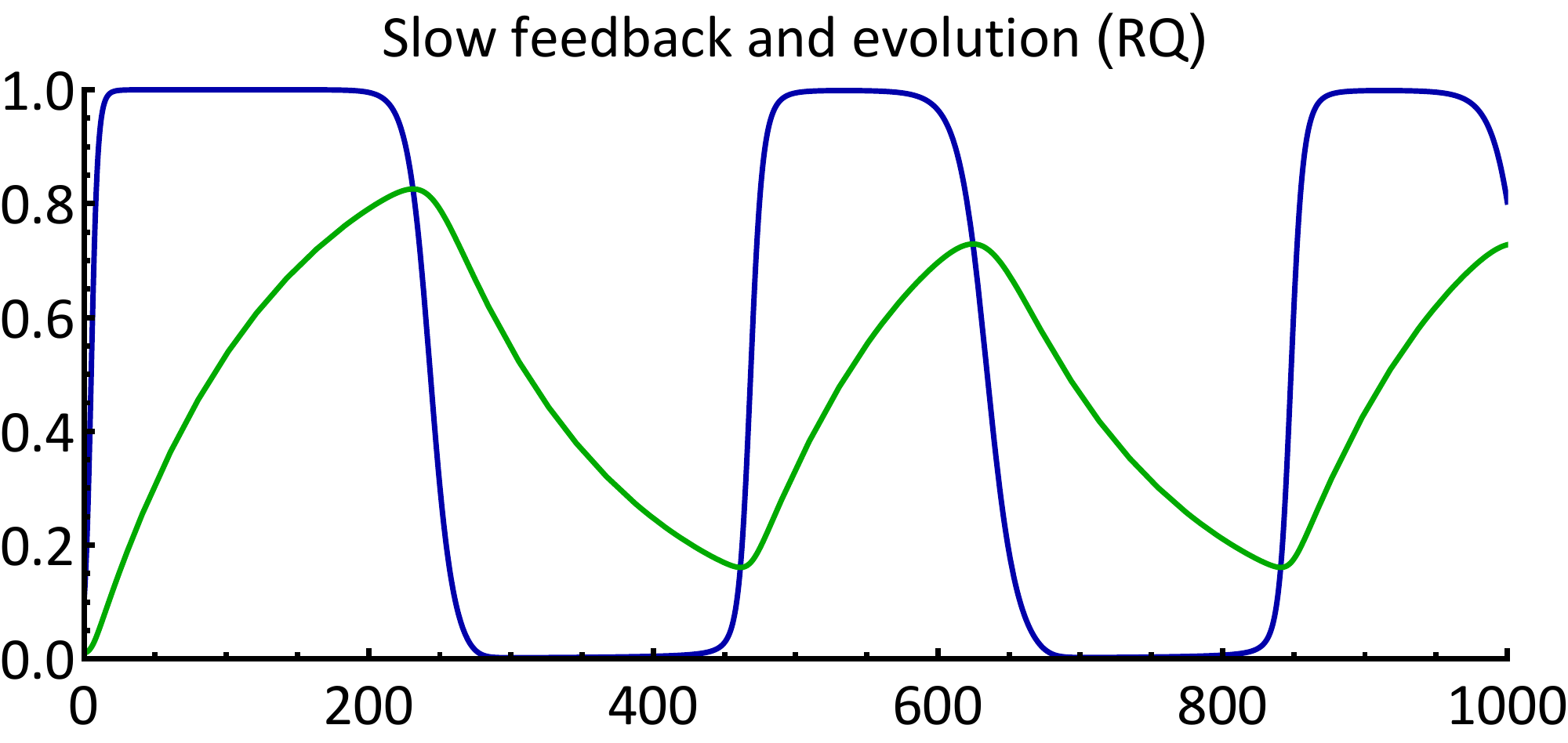} \\
  & \includegraphics[width=.4 \textwidth]{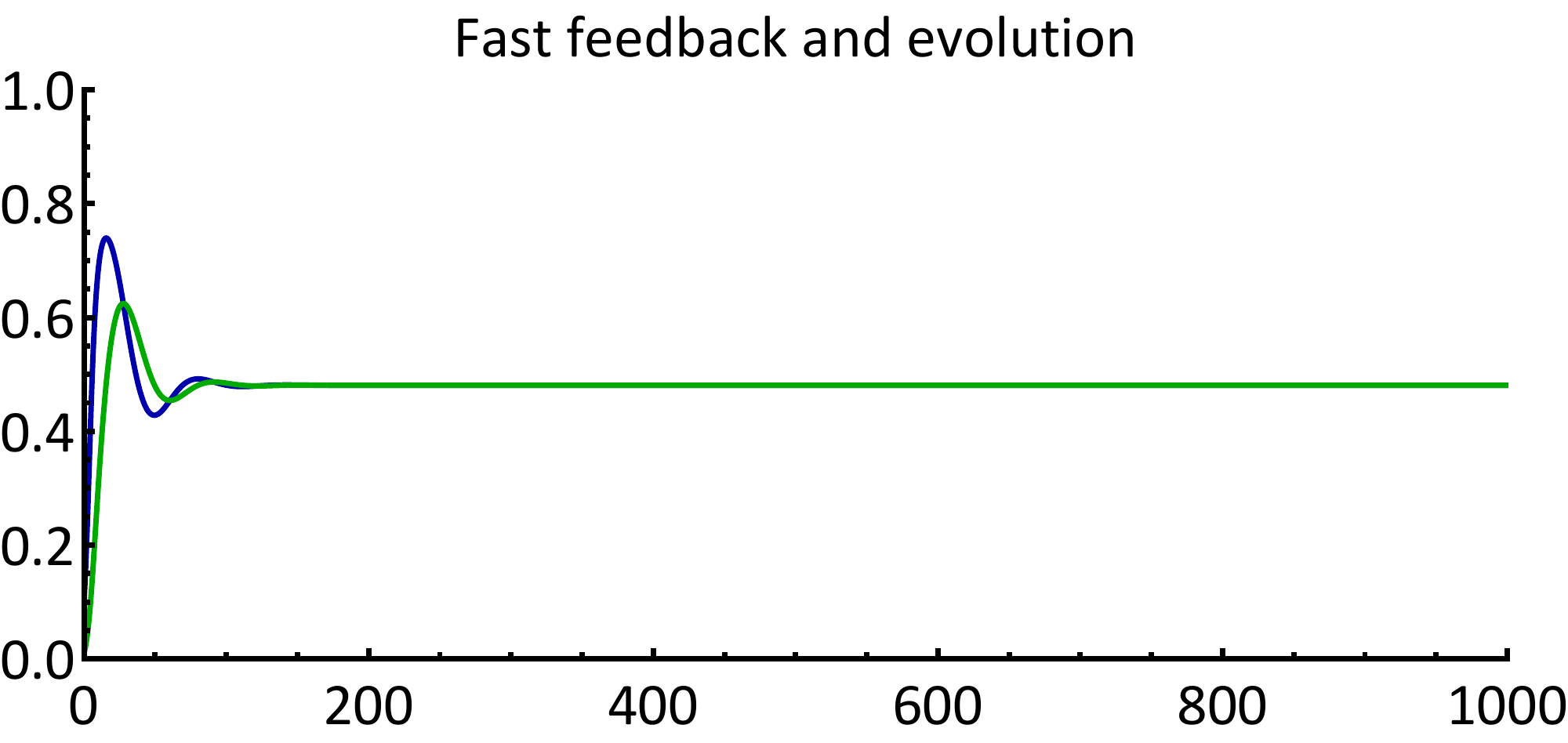} \\
  & \includegraphics[width=.4 \textwidth]{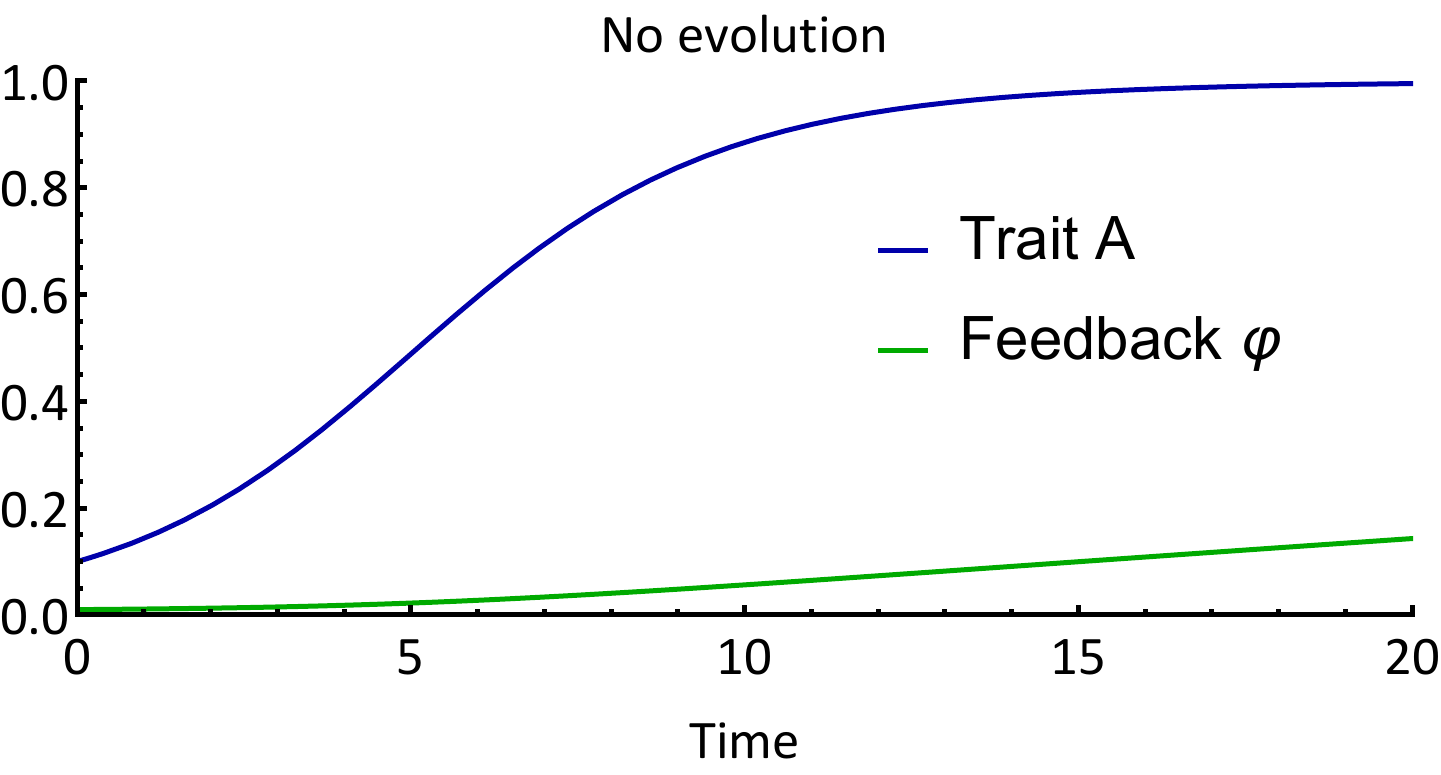}
  \end{tabular}
  \caption{\textbf{Binary trait example} Behaviour of a system with two phenotypes with a fixed population size. A phase plane diagram (left) of trait A and the environmental feedback $\varphi$ shows the possible system behaviour depending on the relative timescales of population dynamics, the feedback and mutations. The density of trait B is the opposite of trait A. Slow feedback and mutations leads to continual cyclic evolutionary dynamics (Red Queen dynamics), almost independent of initial conditions (right top). With fast feedback and mutations an equilibrium is reached (right middle). No evolution (right bottom) leads to the extinction of one of the traits (in this case trait B). We used $k_A = 0.01$ and $k_B = 0.5$ throughout the figure. For the slow feedback and evolution we used $\epsilon_E = 0.008$ and $\epsilon_M = 0.0005$, for the fast feedback and evolution we used $\epsilon_E = 0.1$ and $\epsilon_M = 0.01$ and for no evolution we used $\epsilon_E = 0.01$ and $\epsilon_M = 0$}
  \label{fig:binarySystemR}
\end{figure}

\subsection{One species with two phenotypes and a variable population size} \label{sec:example_twophenovariablepopsize}

In section \ref{sec:twophenovariablepopsize} we give a proof for a system with a variable population size. An example is shown in Figure \ref{fig:binarySystem} in the main text and the equations for the time derivatives of the populations with phenotypes $A$ and $B$ and the inhibitor $\varphi$ for that figure are:
$$
\begin{aligned}
\frac{d A}{d t} & = A \left(\frac{(0.011 + A) \left( 1 - \frac{A + B}{K} \right)}{k_A \left( 1 + \frac{\varphi}{k_B} \right) \left(1+ \frac{0.011+ A}{k_A} \right)} - d \right) + \epsilon_M (B - A) \\
\frac{d B}{d t} & = B \left( 0.5 \left( 1 - \frac{A + B}{K} \right) - d \right) + \epsilon_M (A - B) \\
\frac{d \varphi}{d t} & = \epsilon_E \left( A - \varphi \right)
\end{aligned}
$$

$A$, $B$ and $\varphi$ are all time dependent and therefore short for $A(t)$, $B(t)$ and $\varphi(t)$. $k_A$ and $k_B$ are parameters describing the growth of $A$, $d$ is a death rate and $K$ the carrying capacity. The parameters used for Figure \ref{fig:binarySystem} are $k_A = 0.01$,  $k_B = 0.5$, $d = 0.3$ and $K = 10$. For the slow feedback and evolution we used $\epsilon_E = 0.0005$ and $\epsilon_M = 0.00005$, for the fast feedback and evolution we used $\epsilon_E = 0.1$ and $\epsilon_M = 0.01$ and for no evolution we used $\epsilon_E = 0.01$ and $\epsilon_M = 0$.

\subsection{Multi-varied trait} \label{sec:multivaried}

The equations used for Figure \ref{fig:continuousSystem} in the main text are:

$$
\begin{aligned}
\frac{d u_i}{d t} & = u_i \left( \left[ \frachalf + i \left(\fracthird (M - \frachalf)) - \fractwothird (\varphi - \frachalf) \right)\right] \left( 1 - \frac{U}{K} \right) - d \right)\\
& \qquad\qquad\qquad\qquad\qquad\qquad + \epsilon_M \sum_{j \in I} \frac{1}{n} (u_j - u_i) e^{-10|i-j|} \\
\frac{d \varphi}{d t} & = \epsilon_E \left( M - \varphi \right)
\end{aligned}
$$

The collection of phenotypes $I$ consists of $n$ phenotypes with value $i$ and density $u_i$. The total population size is $U = \sum_i u_i$. The average value of the trait in the population is $M = \frac{\sum_{i \in I} i \cdot u_i}{U}$. $\varphi$ is the level of the feedback compound, $K$ the carrying capacity and $d$ the death rate.

The rate with which mutations change the phenotype decreases exponentially with the difference between the phenotypes. This is to simulate that mutations with small effect are more prevalent than mutations with large effects.

We did a time simulation with different values of the rate of evolution ($\epsilon_M$) and the delay in the slow feedback ($\epsilon_E$). For Figure \ref{fig:continuousSystem} in the main text we used the following parameters: $K = 1$ and $d=0.01$ throughout and for the the RQ dynamics $\epsilon_E = 10^{-4}$ and $\epsilon_M = 10^{-6} $, for the fast feedback and evolution $\epsilon_E = 0.1$ and $\epsilon_M = 0.01$ and when there are no mutations $\epsilon_E = 10^{-4}$ and $\epsilon_M = 0$. We assumed a population to go extinct if the population size was under 0.005.

\subsection{Two feedbacks} \label{sec:example_twofeedbacks}

We simulated a system with 2 traits associated with 2 feedbacks for Figure \ref{fig:TwoFeedbacks} in the main text. The time derivatives of the phenotypes $u_i$, where $ i = \{i_{1}, i_{A}\}$ is a vector of the two phenotypes, are:

$$
\begin{aligned}
\frac{d u_i}{d t} & = u_i \left( \frachalf + \frachalf \left[ i_{1} \left(\fracthird (M_{1} - \frachalf) - \fractwothird (\varphi_{1} - \frachalf) \right) + i_{A} \left(\fracthird (M_{A} - \frachalf) - \fractwothird (\varphi_{A} - \frachalf) \right)\right] \left( 1 - \frac{U}{K} \right) - d \right) \\
&\qquad\qquad\qquad\qquad\qquad\qquad + \epsilon_M \left( \sum_{j \in I | j_{1} = i_{1}} (u_j - u_i) e^{-10|i_{A}-j_{A}|} + \sum_{j \in I | j_{A} = i_{A}} (u_j - u_i) e^{-10|i_{1}-j_{1}|} \right) \\
\frac{d \varphi_{1}}{d t} & = \epsilon_E \left( M_{1} - \varphi_{1} \right) \\
\frac{d \varphi_{A}}{d t} & = \frac{1}{4} \epsilon_E \left( M_{A} - \varphi_{A} \right)
\end{aligned}
$$

The collection of phenotypes $I$ consists of $n$ phenotypes with value $\{i_{1}, i_{A}\}$ and density $u_i$. The total population size is $U = \sum_i u_i$. The average value of the first trait in the population is $M_{1} = \frac{\sum_{i \in I} i_{1} \cdot u_i}{U}$ and of the second trait $M_{A} = \frac{\sum_{i \in I} i_{A} \cdot u_i}{U}$. $\varphi_{1}$ is the level of the feedback compound for the first trait, $\varphi_{A}$ the level of the feedback compound for the second trait, $K$ the carrying capacity and $d$ the death rate.

Since mutations are rare we ignore mutations in both traits at the same time and the rate of mutations in one trait from one phenotype to the other decreases exponentially with increasing difference between the phenotype.

We used the parameters $d = 0.01$, $K = 1$, $\epsilon_E = 0.0001$ and $\epsilon_M = 0.0005$ for Figure \ref{fig:TwoFeedbacks} in the main text.

\subsection{Literature models} \label{sec:literature_models}

We used the competitor-competitor model from \cite{khibnik1997three} where they show RQ dynamics (Fig. 2 in \cite{khibnik1997three}). We changed the model to allow for a polymorphic population, which leads to the following equations:

$$
\begin{aligned}
\frac{d x_i}{d t} & = x_i \left( r_{1,i} - r_2 \sum_i x_i - \sum_j r_{3,i,j} y_j \right) + \epsilon_M \sum_{k \in I} \frac{1}{n} (x_k - x_i) e^{-10|i-k|} \\
\frac{d y_j}{d t} & = y_j \epsilon_E \left( r_{4,j} - r_5 \sum_j y_j - \sum_i r_{6,i,j} x_i \right) + \epsilon_M \sum_{k \in J} \frac{1}{n} (y_k - y_j) e^{-10|j-k|}
\end{aligned}
$$

The parameters are given in \cite{khibnik1997three} Eqs 6 and Fig. 2. $\epsilon_E$ is set to 1 in Main Text Fig. \ref{fig:OtherModels}B and to 0.1 in Fig. \ref{fig:OtherModels}C. The trait values ($i$ and $j$) range from 0.5 to 1.5 and we simulated 100 different phenotypes per species within this range.



\end{document}